\documentclass[superscriptaddress,aps,prl,reprint]{revtex4-2}

\usepackage{amsmath}
\usepackage{amsthm}
\usepackage{dsfont}
\usepackage{mathrsfs}
\usepackage{graphicx}
\usepackage{float}
\usepackage{subfig}
\usepackage{mathtools}
\usepackage{import}
\usepackage{tikz}
\usepackage{ulem}
\usepackage{bbm}
\usepackage{enumerate}

\usepackage{amssymb}
\usepackage[toc,page]{appendix}
\usepackage{braket}
\usepackage{comment}

\usepackage[displaymath, mathlines,running]{lineno}

\usepackage{geometry}
\newtheorem{theorem}{Theorem}

\newtheorem{proposition}[theorem]{Proposition}

\theoremstyle{definition}
\newtheorem{definition}[theorem]{Definition}

\theoremstyle{remark}
\newtheorem{remark}[theorem]{Remark}

\newlength{\ketketwidth}
\newlength{\ketwidth}

\newcommand{\kettstylesep}[3]{
\settowidth{\ketwidth}{$#2\left|#1\right\rangle$}
\settowidth{\ketketwidth}{$#2\left.\left|#1\right\rangle\right\rangle$}
\left|#1\right\rangle#3\hspace{\ketwidth}\hspace{-\ketketwidth}
}
\newcommand{\kett}[1]{
\left.\mathchoice
{\kettstylesep{#1}{\displaystyle}{\hspace{0.3em}}}
{\kettstylesep{#1}{\textstyle}{\hspace{0.3em}}}
{\kettstylesep{#1}{\scriptstyle}{\hspace{0.3em}}}
{\kettstylesep{#1}{\scriptscriptstyle}{\hspace{0.25em}}}
\right\rangle
}

\begin{document}
\title{Quantum mechanics based on real numbers: A consistent description}
\author{Pedro Barrios Hita}
\affiliation{Institute of Software Technology, German Aerospace Center (DLR), Sankt Augustin, Germany}
\affiliation{Heinrich Heine University D\"{u}sseldorf, Faculty of Mathematics and Natural Sciences, Institute for Theoretical Physics III}
\email{Pedro.Barrios@dlr.de}
\author{Anton Trushechkin}
\affiliation{Heinrich Heine University D\"{u}sseldorf, Faculty of Mathematics and Natural Sciences, Institute for Theoretical Physics III}
\author{Hermann Kampermann}
\affiliation{Heinrich Heine University D\"{u}sseldorf, Faculty of Mathematics and Natural Sciences, Institute for Theoretical Physics III}
\author{Michael Epping}
\affiliation{Institute of Software Technology, German Aerospace Center (DLR), Sankt Augustin, Germany}
\author{Dagmar Bruß}
\affiliation{Heinrich Heine University D\"{u}sseldorf, Faculty of Mathematics and Natural Sciences, Institute for Theoretical Physics III}
\date{\today}

\begin{abstract} 
Complex numbers play a crucial role in quantum mechanics. However, their necessity remains debated: whether they are fundamental or merely convenient. 
Recently, it was shown that any real-number quantum theory satisfying
certain postulates can be falsified with multipartite experiments. In this article we show that a physically motivated postulate about composite quantum systems allows to construct quantum mechanics based on real numbers that reproduces predictions for all multipartite quantum experiments. Thus, we
argue that real-valued quantum mechanics cannot be falsified, and therefore the use of
complex numbers is a matter of convenience. \newline

\end{abstract}
\maketitle

\begin{figure}[t]
\centering
\includegraphics[scale=.7]{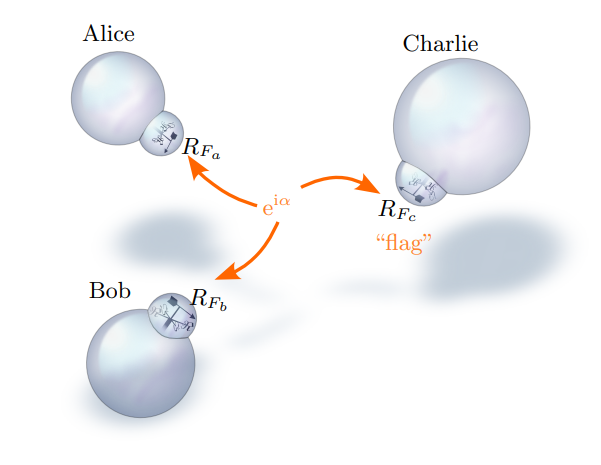}
\caption{A tripartite quantum system. The small spheres represent local two-dimensional flags, which are attached to each party to emulate the complex structure. There, a global phase factor $e^{i\alpha}$ can be arbitrarily split between the subsystems. The usual tensor product of real Hilbert spaces does not induce this ambiguity. However, the introduction of an equivalence relation on the multipartite flag state implements the same ambiguity: A global phase $\mathrm{e}^{i\alpha}$ can be arbitrarily split between the local flag rotations $R_{F_a}$, $R_{F_b}$, and $R_{F_c}$. This formulation of quantum mechanics reproduces the predictions for all multipartite quantum experiments using only real numbers.}
\label{fig:generalscenario}
\end{figure}

Quantum mechanical predictions and measurement results are real-valued, whereas the abstract quantum mechanical formalism typically relies on using complex numbers. It is therefore of fundamental interest to explore the implications of formulating quantum mechanics without complex numbers.

In addition to this intuitive argument, there are several mathematical ones that support the real formulation of quantum mechanics. 
Operations like
time-reversal
that are antilinear in the complex formulation become
linear in the real formulation \cite{Dyson_1962}.
Additionally, quantum logic’s algebraic structures can consistently be mapped to real Hilbert spaces \cite{logicqm}. Furthermore, key properties like the ability to construct orthogonal bases and to define norms are preserved in real Hilbert spaces, suggesting that these spaces offer a consistent framework \cite{foundqm,geoqm}. 
It is speculated that any fundamental physical theory can be constructed on any number field and the fundamental physical laws are invariant under the change of the number field \cite{Volovich2010}, see, e.g., works on $p$-adic quantum mechanics \cite{VVZ,Dragovich2017}.

The topic of a theory called ``real quantum mechanics'' (RQM) has been extensively studied. Two classes of RQM can be distinguished. First, the constructions from Refs.~\cite{Stueckelberg, Myrheim, McKague2009, Koh2018, Aleksandrova2013}, which successfully predict experimental outcome statistics by including a binary real-valued system, often called a rebit in the literature \cite{Wootters2012}. These constructions are therefore equivalent to standard quantum mechanics, referred to throughout this text as ``complex quantum mechanics".
However, the rebit is ``universal'' in the sense of being accessible to all subsystems of a global composite state, even under space-like separation. This property was interpreted as non-locality of the rebit. Second, the models considered in Refs. \cite{barret,hardy,adri}, which use the tensor product postulate for composite systems. These models are mathematically inequivalent to standard quantum mechanics.

Recently, it has been shown in the context of multipartite Bell-type experiments~\cite{Renou2021,Renou2025,gisin,elliot} that the second class of formulations of quantum mechanics using real numbers,  would be incompatible with the experimental outcome statistics, thus rendering them subject to experimental falsification.
The corresponding experiments were carried out \cite{Chen_2022,Wu_2022,PhysRevLett.128.040402}, and their results were indeed inconsistent with the second class of formulations of real-number quantum mechanics. A crucial postulate in the above articles is that the Hilbert space of a composite system is given by the tensor product of the Hilbert spaces of the individual subsystems. However, this mathematical postulate, though widely adopted, is too restrictive for defining a real-vector-space formulation of standard quantum mechanics. Here we will show that postulating a more fundamental physical requirement instead, leads to a real-valued theory which is equivalent  to complex  quantum mechanics. We postulate that, given two separate subsystems, an operation acting only on one subsystem does not have a measurable effect on the other one. In other words, a local operator must  act trivially on all other subsystems. Note that the tensor product structure of complex quantum mechanics  naturally fulfills our postulate (see Section E of the Supplementary Material).

In this article, we propose a real-valued construction for composite quantum systems. Note that its form is similar to the ones from the literature \cite{Stueckelberg, Myrheim, McKague2009, Aleksandrova2013}. However, we 
introduce a mathematical structure
and show that it does fulfill the above requirement of locality. 
Our construction is based on the notion of the composite Hilbert space as a real quotient space. Thus, it reflects certain equivalences of state vectors that are automatically contained in a complex formulation of quantum mechanics.
As our construction is a one-to-one mapping of complex quantum mechanics to a real description, it is compatible with the outcomes of
all quantum experiments, including multipartite Bell-type scenarios.

We propose the following postulates for both complex  and real quantum mechanics.
(P1) \textit{For every physical system $A$, there is a complex (real) Hilbert space $\mathcal{H}_A$ in which the state of the former is represented by a normalized vector $\ket{\phi}$, i.e., $\langle\phi\vert \phi\rangle=1$.}
(P2) \textit{A measurement $\Pi$ of 
system $A$ corresponds to an ensemble of positive semi-definite operators   
$\{\Pi_r\}_r$ indexed by the measurement outcome $r$ and acting on $\mathcal{H}_A$, obeying $\sum_r\Pi_r=\mathds{1}_A$, where $\mathds{1}_A$ denotes the identity.
}
(P3) \textit{Born Rule: If $\Pi$ is measured while $A$ is in the state $\ket{\phi}$, the probability of measuring 
outcome $r$ is given by $P(r)=\bra{\phi}\Pi_r\ket{\phi}$.}
(P4) \textit{Given two subsystems $A$ and $B$, the operators used to describe measurements or transformations of system $A$, act trivially on $B$, and vice versa. In particular, this means that any two local operators $O_A$ and $O_B$, acting on subsystems A and B, respectively, commute (i.e., $[O_A, O_B]=0$)} \cite{Peres99}.
Note that postulates (P1) -- (P3) are analogous to the ones in Ref.~\cite{Renou2021}, and only (P4) is different.
Further details on the definition of postulate (P4) can be        found in Section E of the Supplementary Material \cite{supp_mat}.

We will now show how these postulates allow a description of quantum mechanics using only real numbers. We start in complex quantum mechanics with a single quantum state.
An abstract state $\ket{\psi}$ of a $d$-dimensional quantum system  can be represented as \mbox{$\ket{\psi}= (c_1, c_2,...,c_d)^T$} in a fixed basis, with $c_i\in\mathbb{C}$. It can also be represented as $\ket{\psi}= (\mathrm{Re}(c_1),...,\mathrm{Re}(c_d), \mathrm{Im}(c_1),...,\mathrm{Im}(c_d))^T$. We will denote in the following $\mathrm{Re}(\ket{\psi}):=(\mathrm{Re}(c_1),...,\mathrm{Re}(c_d))^T$ and
$\mathrm{Im}(\ket{\psi}):=(\mathrm{Im}(c_1),...,\mathrm{Im}(c_d))^T$.
A mapping from the complex to the real representation of a (pure) state $\ket{\psi}\in\mathbb{C}^d$ (see also, e.g., Refs.~\cite{Myrheim,McKague2009,Aleksandrova2013,Renou2021}) can
thus be written in the convenient form
\begin{equation}
\label{eq:map1partystates}
\mathcal{S}:\ket{\psi}\mapsto\mathrm{Re}(\ket{\psi})\otimes\ket{0}_{F}+\mathrm{Im}(\ket{\psi})\otimes\ket{1}_{F} =: \ket{\widetilde\psi}\ ,
\end{equation}
where the subscript $F$ stands for ``flag'':
the two  real degrees of freedom of the flag indicate the real and imaginary part of  the original state vector.
Here the real Hilbert space has the form $\mathcal H=\mathbb{R}^d\otimes\mathbb R^2_F$.
We point out that even though  Eq.~(\ref{eq:map1partystates}) 
formally resembles the structure  of an entangled state, $\ket{\widetilde\psi}$ is a \textit{single} quantum state and thus is not entangled. Throughout this paper we will denote the real representation of states and
operators with a tilde. The bra vector $\bra{\widetilde\psi}$ is defined as $\ket{\widetilde\psi}^T$. As $\braket{\widetilde\psi|
\widetilde\psi}= \braket{\psi|\psi}$ holds, normalisation of 
$\ket{\psi}$ ensures normalisation of $\ket{\widetilde\psi}$, see (P1).

In quantum mechanics, a normalised state vector $\ket{\psi}$ is physically indistinguishable from, and thus equivalent to, $e^{i\alpha}\ket{\psi}$, with 
$\alpha\in\mathbb{R}$; i.e., a global phase is undetectable.  A global phase translates in the real representation, see Eq.~(\ref{eq:map1partystates}), 
to the rotation $R_F(\alpha)$  by the angle $\alpha$ acting on the flag:
\begin{align}
\label{eq:map1phasestates}
\mathcal{S}\colon e^{i\alpha}\ket{\psi} &\mapsto
\text{Re}(\ket{\psi}) \otimes
(\cos{\alpha}\ket{0}_F +
\sin{\alpha}\ket{1}_F) \nonumber \\
&+ \text{Im}(\ket{\psi})  \otimes
(-\sin{\alpha}\ket{0}_F +
\cos{\alpha}\ket{1}_F) \nonumber \\
& =   R_F(\alpha) \ \mathcal{S}(\ket{\psi})
\ .
\end{align}
Thus, the global phase or U$(1)$-ambiguity in the definition of a complex quantum state vector corresponds to an SO$(2)$-ambiguity in the real formulation. The following example shows that real quantum mechanics  
has certain features that are quite different from complex quantum mechanics: For $\alpha = \pi/2$, the real representations of $i\ket{\psi}$
and $\ket{\psi}$ are orthogonal to each other, even though these states are physically indistinguishable.

The physical indistinguishability of two real state vectors $\ket{\widetilde\psi}$ and $\ket{\widetilde\phi}=R_F(\alpha)\ket{\widetilde\psi}$ is 
guaranteed iff 
the real measurement operators from (P2) lead to identical outcome probabilities, see (P3).
Let us denote the map from a complex operator to its real representation by 
$\mathcal T$, and the real image of a complex operator
$\Pi$ by $\widetilde\Pi={\mathcal T}(\Pi)$. 
We require
\begin{equation}
\braket{\widetilde\psi|
\widetilde\Pi|\widetilde\psi}
=
\braket{\widetilde\phi|
\widetilde\Pi|\widetilde\phi}
=
\braket{\widetilde\psi|
R_F(\alpha)^\dagger\,\widetilde\Pi \,R_F(\alpha)|\widetilde\psi} \ 
\end{equation}
for any real vector $\ket{\widetilde\psi}$ of the form given in Eq.~(\ref{eq:map1partystates}) and for arbitrary $\alpha\in\mathbb{R}$. 
Hence, $\widetilde\Pi$ must commute with $R_F(\alpha)$ for all $\alpha$.
Also requiring that $\mathcal S(\Pi\ket\psi)=\mathcal T(\Pi)\mathcal S(\ket\psi)$  implies that real operators have the form (for details see Section A of the Supplementary Material \cite{supp_mat})
\begin{equation}
\label{eq:RPi}
{\mathcal T}(\Pi)=\mathrm{Re(\Pi)}\otimes I_F + \mathrm{Im(\Pi)}\otimes J_F ,
\end{equation}
where
\begin{equation}
\label{eq:defJ}
I_F=\begin{pmatrix}
1&0\\0&1
\end{pmatrix}
,
\qquad
J_F=\begin{pmatrix}
0&-1\\1&0
\end{pmatrix} ,
\end{equation}
and $\mathrm{Re(\Pi)}$ ($\mathrm{Im(\Pi)}$) denotes the matrix consisting of the real (imaginary) parts of the elements of $\Pi$ in the basis used to define $\mathrm{Re}(\ket\psi)$ and $\mathrm{Im}(\ket\psi)$. 
Note that the restriction for the  observables  to the form given in Eq.~\eqref{eq:RPi} forbids nontrivial measurement of the flag subsystem only (e.g., the projective measurement $\{\mathds{1}\otimes\ket0\bra0_F, \mathds{1}\otimes\ket1\bra1_F\}$). In this sense, the flag is not a directly accessible degree of freedom.

An alternative derivation of the map $\mathcal{T}$ given in Eq.~(\ref{eq:RPi}) based on a quotient space argument is given in Section C and additional properties  of $\mathcal{T}$ are in Section D of the Supplementary Material \cite{supp_mat}. 

The real representation in Eq.~(\ref{eq:RPi})
preserves the algebraic structure, i.e., $\mathcal{T}(A_1 A_2)=\mathcal{T}(A_1)\mathcal{T}(A_2)$ for any (complex) observables $A_1$ and $A_2$, and satisfies
\begin{equation}
\mathcal{S}(\ket{\psi})^T {\mathcal{T}}(A)\mathcal{S}(\ket{\psi}) = \bra{\psi} A \ket{\psi}, 		\label{eq:expvalues}
\end{equation}
i.e., the maps $\mathcal S$ and $\mathcal T$ for vectors and operators leave all expectation values invariant.

Let us now derive the description of composite quantum systems in a real vector space. Here we analyse for simplicity the case of two two-dimensional subsystems
labeled $A$ and $B$.   The structure of the states of composite systems with more than two subsystems and arbitrary dimension is detailed in Section B of the Supplementary Material \cite{supp_mat}.

A straightforward attempt for the mapping from the complex to a real vector space would be to extend the map  $\mathcal{S}$ to the tensor product 
$\mathcal{S}^{\otimes 2}$, i.e.,
$\ket{\psi_A} \otimes \ket{\psi_B}\mapsto
\mathcal{S}(\ket{\psi_A}) \otimes  \mathcal{S}(\ket{\psi_B})$. However, it is obvious that the dimensions would not match: a vector in 
$\mathbb{C}^2_A\otimes \mathbb{C}^2_B$ has $2\cdot2\cdot2=8$ real parameters, while a vector in 
$\mathbb{R}^4_A\otimes \mathbb{R}^4_B$ has $4\cdot 4=16$ parameters. In addition, remember the following important property of quantum mechanics: Given two independent subsystems, each of them can have an undetectable phase, see above. Due to linearity in each subsystem, this can also be interpreted as a global phase which is equal to the product of the individual phases. Vice versa, a global phase can be split up between the subsystems in an arbitrary way. 
Translating this property to a real representation,
as schematically illustrated in 
Figure~\ref{fig:generalscenario}, is the aim for the remainder of this article. The fact that the phase of a tensor product can be attributed to any of the subsystems is used in some quantum algorithms and has been referred to in the literature as ``phase kickback'' \cite{kickback1,kickback2,kickback3,q_alg_rev}.

As an example for this phase ambiguity in tensor products, it holds that $\ket {0}_A\otimes\ket{1}_B = i\ket{0}_A\otimes(-i\ket{1}_B)$. However, the images of the left-hand side and the right-hand side of this equality, when applying $\mathcal{S}$ to each qubit individually, are different:  $\mathcal{S}(\ket {0}_A)\otimes\mathcal{S}(\ket {1}_B)=\ket{0}_A\otimes \ket{0}_{F_a}\otimes{\ket1}_B\otimes{\ket0}_{F_b}$,
which is not equal to
$\mathcal{S}(i\ket {0}_A)\otimes\mathcal{S}(-i\ket {1}_B)=-\ket{0}_A\otimes \ket{1}_{F_a}\otimes{\ket1}_B\otimes{\ket1}_{F_b}$. Thus, tensoring the map 
$\mathcal{S}$ does not lead to a well-defined map. This is sketched in Figure~\ref{fig:balloons}. There and in the following we will group the flag states together, as indicated by the labels $F_{a(b)}$. We will also omit the tensor product symbols between the systems $A$ and $B$, and between the flags.

\begin{figure}[tp]
\centering
\def\svgwidth{\linewidth}
\includegraphics[scale=.65]{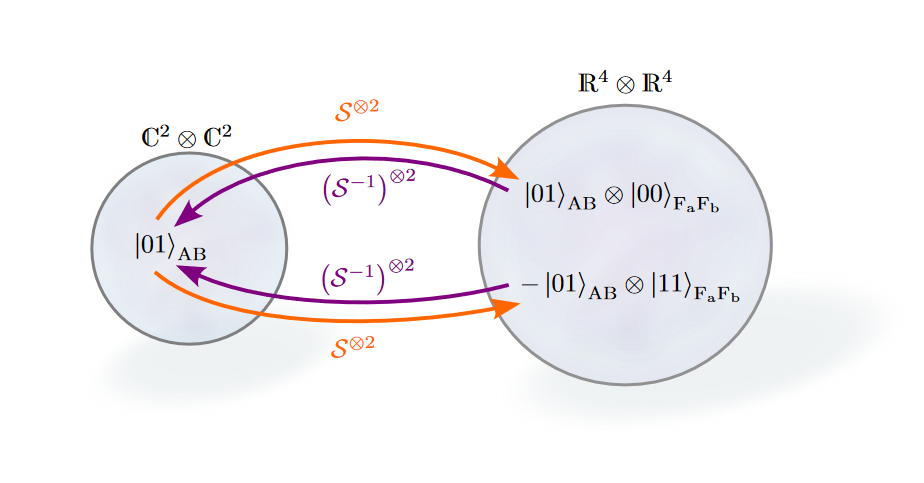}
\caption{Visualisation of the ill-defined ``map'' $\mathcal{S}^{\otimes 2}$
(orange arrows). Since $\ket{0}_A\otimes{\ket1}_B=i\ket0_A\otimes (-i)\ket1_B$, one element in the domain would have more than one image.
The tensored inverse map 
$(\mathcal{S}^{-1})^{\otimes 2} $ (purple arrows) is well-defined but not invertible: Our example shows two elements in the domain which have the same image, and thus the kernel is non-trivial.}
\label{fig:balloons}
\end{figure}

Therefore, one has to follow a different approach. Instead of extending $\mathcal{S}$ to the tensor product, we extend its inverse $\mathcal S^{-1}$, which is well-defined. The states $\ket{01}_{AB}\otimes\ket{00}_{F_aF_b}$ and $-\ket{01}_{AB}\otimes\ket{11}_{F_aF_b}$ have the same image under the map $(\mathcal S^{-1})^{\otimes 2}$
(see Figure~\ref{fig:balloons}). These vectors correspond to the same physical state and are thus equivalent.
Two vectors are equivalent if their difference belongs to the kernel ${\rm ker}[(\mathcal S^{-1})^{\otimes 2}]$ of the map $(\mathcal S^{-1})^{\otimes 2}$. It is easy to check that  $\mathrm{ker}[(\mathcal S^{-1})^{\otimes 2}]$ is spanned by the vectors of the form $\ket{\psi'}_{AB}\otimes(\ket{00}_{F_aF_b}+\ket{11}_{F_aF_b})$ and $\ket{\psi'}_{AB}\otimes(\ket{01}_{F_aF_b}-\ket{10}_{F_aF_b})$, where $\ket{\psi'}_{AB}\in\mathbb R^2\otimes\mathbb R^2$.

Consider then, instead of the tensor product 
$\mathbb{R}^4_A\otimes \mathbb{R}^4_B$,
the quotient space $\mathbb{R}^4_A\otimes \mathbb{R}^4_B/{\rm ker}[(\mathcal S^{-1})^{\otimes 2}]$, i.e., the space of the corresponding equivalence classes. 
This leads to a one-to-one map between the complex and the real representation of quantum mechanics.
Then the equivalence classes $\kett{\cdot}_{F}$ associated with the flag basis states for the composite system are

\begin{equation}
\label{eq:FlagEquivClasses}
\begin{aligned}
\kett{0}_{F}=
&\ket{\psi_{\mathrm{even}}^{(2)}}_{F_aF_b}\\
+&{\rm span}\{\ket{00}_{F_aF_b}+\ket{11}_{F_aF_b},
\:
\ket{01}_{F_aF_b}-\ket{10}_{F_aF_b}\}\\
\kett{1}_{F}=&\ket{\psi_{\mathrm{odd}}^{(2)}}_{F_aF_b}\\
+&{\rm span}\{\ket{00}_{F_aF_b}+\ket{11}_{F_aF_b},
\:
\ket{01}_{F_aF_b}-\ket{10}_{F_aF_b}\},
\end{aligned}
\end{equation}
where $\ket{\psi_{\mathrm{even}}^{(2)}}_{F_aF_b}$ and $\ket{\psi_{\mathrm{odd}}^{(2)}}_{F_aF_b}$ are the canonical representatives of these equivalence classes (i.e., vectors orthogonal to both $\ket{00}_{F_aF_b}+\ket{11}_{F_aF_b}$ and $\ket{01}_{F_aF_b}-\ket{10}_{F_aF_b}$). 
Explicitly, 
\begin{equation}
\begin{aligned}
&\ket{\psi_{\mathrm{even}}^{(2)}}=\frac{1}{\sqrt{2}}(\ket{00}_{F_aF_b}-\ket{11}_{F_aF_b})\qquad \text{and}\\ &\ket{\psi_{\mathrm{odd}}^{(2)}}=\frac{1}{\sqrt{2}}(\ket{01}_{F_aF_b}+\ket{10}_{F_aF_b}).
\label{eq:evenodd}
\end{aligned}
\end{equation}
Here the label even (odd) indicates an even (odd) number of entries 1 in each term of the superposition.  Note that the vectors in Eq.~(\ref{eq:evenodd}) were also employed in Ref.~\cite{McKague2009}. Here we give a constructive explanation of them: They are the canonical representatives of equivalence classes.
A related approach had been investigated in Ref.~\cite{Myrheim}; we refer to Section G of the Supplementary Material \cite{supp_mat} for a detailed comparison with our construction.

We point out that, even though the states in Eq.~(\ref{eq:evenodd}) formally resemble entanglement, they describe separable states.
We define a pure state in real quantum mechanics as separable if and only if it is either a product state or equivalent to a product state. 
Otherwise, it is entangled.
Thus, the states $\ket{\psi_{\mathrm{even}}^{(2)}}$ and $\ket{\psi_{\mathrm{odd}}^{(2)}}$ are in a superposition without being entangled since they are equivalent to the product states $\ket{00}_{F_aF_b}$ 
and $\ket{01}_{F_aF_b}$, respectively. This is reminiscent of the symmetrised wavefunction of indistinguishable particles, which is in a superposition that also does not imply entanglement \cite{Eckert_2002}.

Thus,
the elements of the whole quotient space $\mathbb R_A^4\otimes\mathbb R_B^4/{\rm ker}[(\mathcal S^{-1})^{\otimes2}]$ can be conveniently expressed  as Eq.~(\ref{eq:map1partystates}), with $\ket{0}_{F}$ and $\ket{1}_{F}$ now replaced by the flag equivalence classes $\kett{0}_F$ and $\kett{1}_F$ or equivalently by the canonical representatives $\ket{\psi_{\mathrm{even}}^{(2)}}$ and $\ket{\psi_{\mathrm{odd}}^{(2)}}$, 
see Section A of the Supplementary Material \cite{supp_mat} for details. Since we did not specify the ``nature'' of the vectors $\ket0_{F}$ and $\ket1_{F}$, Eq.~(\ref{eq:map1partystates}) and all its consequences, in particular Eq.~(\ref{eq:expvalues}), are valid also for composite systems, thus proving that the presented real-valued formulation of quantum mechanics reproduces the predictions of the complex quantum mechanics for all experiments, also multipartite ones. An explicit real-valued formulation of the experiment suggested in Ref.~\cite{Renou2021} is given in Section H of the Supplementary Material \cite{supp_mat}.

Remember that our central argument is the physically motivated postulate (P4). In order to show how our construction guarantees the local action of operators, 
consider the tensor product of two operators  of the form given in Eq.~(\ref{eq:RPi}): Tensor products $I_{F_a}\otimes I_{F_b}$, $I_{F_a}\otimes J_{F_b}$, $J_{F_a}\otimes I_{F_b}$, and $J_{F_a}\otimes J_{F_b}$ of the flag operators appear. As can be easily seen, $I_{F_a}\otimes I_{F_b}$ and $-J_{F_a}\otimes J_{F_b}$ act as the identity operator $I_F$ on the basis states (equivalence classes) $\kett{0}_F$ and $\kett{1}_F$, see Eq.~(\ref{eq:FlagEquivClasses}), of the flag of the composite system. Analogously, $I_{F_a}\otimes J_{F_b}$ and $J_{F_a}\otimes I_{F_b}$ act as $J_F$ on these basis states (see also Eq.~\eqref{eq:IJN+M} of~\cite{supp_mat}).\footnote{Also these four product operators map canonical representatives into canonical representatives.} Hence, the embedding of the local actions on subsystems into the operators for the composite system is well-defined. The postulate of local action (P4) is satisfied because the underlying structure is still the tensor product of the Hilbert spaces, though with an additional equivalence relation.  (The satisfaction of  (P4) in real and complex quantum mechanics is analyzed in more detail in Section E of the Supplementary Material \cite{supp_mat}.) The reduced density operator in the real-valued formalism is still given via the (real) partial trace, which is derived in Section F of the Supplementary Material \cite{supp_mat}. 

In summary, we propose a formulation of quantum mechanics based on real numbers, which leads to exactly the same predictions for experimental outcomes as complex quantum mechanics (see Eq.~\eqref{eq:expvalues} and discussion below). Thus, we conclude that both theories are equivalent in this sense.
In order to arrive at this formulation, we show that the tensor product postulate does not capture the ambiguity in assigning the global phase of a product state to its subsystems \cite{kickback1,kickback2,kickback3,q_alg_rev}.
Instead, we argue that a postulate based on {\it physical} properties for composite systems -- namely  that a local operation on one subsystem acts trivially on the other -- is well motivated and does lead to a theory that
corresponds to complex quantum mechanics.
The above-mentioned phase ambiguity is ensured by employing the quotient space of the real tensor product. Moreover, in Section A of the Supplementary Material we show that, under certain natural conditions related to representations of SO(2) rotations, our construction 
not only fulfills the required physical postulate but is also unique, up to Hilbert space isomorphisms.

Note that a similar real representation has been discussed in the literature before \cite{Stueckelberg, Myrheim, McKague2009, Aleksandrova2013} -- however, without linking it to the role of the tensor product postulate.
We mention that a tensor product in complex quantum mechanics can also be formulated as a quotient space \cite{elkies_tensor_2010,Holevo2002} and fulfills our physically motivated postulate. Let us point out other recent papers about equivalent real-valued formulations of quantum mechanics \cite{volovichkahler,mischa}. 

As stated in Ref. \cite{Spekkens}, the results of Refs. \cite{Renou2021,Renou2025,gisin,elliot} are best understood not as a claim about the impossibility of a real-valued Hilbert space representation of quantum theory in principle, ``but rather as a claim about the possibility of experimentally
adjudicating between standard quantum theory and an alternative theory--a foil theory--known
as real-amplitude quantum theory'' (which includes the real-valued version of the tensor product postulate).

In conclusion, complex numbers are not necessary to describe quantum mechanics -- but they are certainly very useful. 

We thank David Edward Bruschi, Markus Grassl, Alexander Holevo, Alexander Kegeles, Yien Liang, Peter Ken Schuhmacher, and Igor Volovich for stimulating discussions during the course of this work. We also thank an anonymous referee for constructive criticism that helped to improve clarity of this manuscript.
This project was made possible by the DLR Quantum Computing Initiative and the Federal Ministry for Economic Affairs and Climate Action. HK and DB acknowledge support by
Deutsche Forschungsgemeinschaft
(DFG, German Research Foundation) under Germany’s
Excellence Strategy -- Cluster of Excellence Matter and
Light for Quantum Computing (ML4Q) EXC 2004/1 --
390534769. AT, HK and DB acknowledge support by
the German Ministry of Education and Research
(Project QuKuK, BMBF Grant No. 16KIS1618K).

\newpage
\bibliography{real_qm_paladins.bib}

\clearpage
\onecolumngrid
\setcounter{section}{0}
\renewcommand{\thesection}{\Alph{section}}
\setcounter{equation}{0}
\renewcommand{\theequation}{S\arabic{equation}}
\setcounter{figure}{0}
\renewcommand{\thefigure}{S\arabic{figure}}
\setcounter{table}{0}
\renewcommand{\thetable}{S\arabic{table}}

\begin{center}
{\large\bfseries Supplemental Material for\\[0.5em]
``Quantum mechanics based on real numbers: A consistent description''}

\vspace{1em}
Pedro Barrios Hita,$^{1,2}$ Anton Trushechkin,$^{2}$ Hermann Kampermann,$^{2}$ Michael Epping,$^{1}$ and Dagmar Bru\ss$^{2}$

\vspace{0.5em}
{\small
$^{1}$Institute of Software Technology, German Aerospace Center (DLR), Sankt Augustin, Germany\\
$^{2}$Heinrich Heine University D\"usseldorf, Faculty of Mathematics and Natural Sciences, Institute for Theoretical Physics III
}
\end{center}

\vspace{1em}
\section{Uniqueness of the mapping of complex-valued to real-valued quantum mechanics}
\label{ap:uniqueness}

In this section, we prove that the suggested construction of quantum mechanics based on real numbers (some elements of which appeared in earlier papers), see equations~(1) and~(4), is unique under certain assumptions. Our construction is based on the isomorphism between the groups U(1) and SO(2) and the requirement that the U(1)-ambiguity in the definition of a complex (pure) quantum state transforms into the SO(2)-ambiguity for the real case.

\begin{proposition}
\label{prop:uniquestates}
Consider a map $\mathcal S$ from $\mathbb C^d$ to some (unknown) real Hilbert space $\mathcal H$ satisfying the following properties:
\begin{enumerate}[(i)]
\item $\mathcal S$ is real-linear, i.e.,
\begin{equation}
\mathcal S(a\ket\phi+b\ket\psi)=a\mathcal S(\ket\phi)+b\mathcal S(\ket\psi)
\end{equation}
for arbitrary $\ket\psi,\ket\phi\in\mathbb C^d$ and $a,b\in\mathbb R$.

\item A phase factor is mapped to the corresponding $\rm SO(2)$ rotation:
\begin{equation}
\mathcal S(e^{i\alpha}\ket\psi)=g(R_{\alpha})\mathcal S(\ket\psi),
\end{equation}
where $R_\alpha\in {\rm SO}(2)$ is the rotation on the angle $\alpha\in\mathbb R$ and $g$ is a representation of $\rm SO(2)$.

\item The map $\mathcal S$ preserves length and orthogonality: If $\braket{\phi|\psi}$ is zero or one, then $\mathcal S(\ket\phi)^T\mathcal S(\ket\psi)$ is also zero or one, respectively.

\item The range of $\mathcal S$ is the whole $\mathcal H$.
\end{enumerate}

Then $\mathcal H\cong\mathbb R^d\otimes\mathbb R^2$ (here and in the following, $\cong$ denotes the Hilbert space isomorphism) and there exists an orthonormal basis $\{\ket{e_j}\}_{j=1}^d$ of $\mathbb C^d$ and an orthonormal basis $\{\ket{\widetilde e_j}\otimes\ket0_F,\ket{\widetilde e_j}\otimes\ket1_F\}_{j=1}^d$ of $\mathcal H$ such that 
\begin{equation}
\label{eq:Sdetailed}
\mathcal S\left(
\sum_{j=1}^d
c_j\ket{e_j}
\right)
=\sum_{j=1}^d
[{\rm Re}(c_j)
\ket{\widetilde e_j}
\otimes\ket0_F
+
{\rm Im}(c_j)
\ket{\widetilde e_j}
\otimes\ket1_F]
\end{equation}
for arbitrary complex $c_j$,
or, in short,
\begin{equation}
\label{eq:S}
\mathcal S(\ket\psi)=
{\rm Re}(\ket{\psi})\otimes\ket0_F
+
{\rm Im}(\ket{\psi})\otimes\ket1_F.    
\end{equation}
\end{proposition}

Note that condition (iii) means that the map $\mathcal S$ does not ``erase information'': States perfectly distinguishable in the complex quantum mechanics must be perfectly distinguishable also in the real case. Condition (iv) is technical. If the image of $\mathcal S$ is only a subspace of $\mathcal H$, then it has no sense to consider this part and we can actually restrict $\mathcal H$ to this subspace. Thus, (iv) does dot actually restrict generality, but says that only the image of $\mathcal S$ is relevant. Also, in the main text, we used the tilde for the real images of complex vectors. In Eq.~\eqref{eq:Sdetailed}, it is used in a slightly different sense: It denotes only the ``main'' (without a flag) part of the real image of the corresponding complex basis vectors.

\begin{proof}
Let us start the proof with a derivation of two  equalities. By linearity of $\mathcal S$, we have
\begin{equation}
\begin{split}
g(R_{\alpha})\mathcal S(\ket\psi)\overset{(ii)}{=}\mathcal S(e^{i\alpha}\ket\psi)
&\overset{(i)}{=}
\cos\alpha\,
\mathcal S(\ket\psi)
+
\sin\alpha\,
\mathcal S(i\ket\psi)
\\
&\overset{(ii)}{=}
\cos\alpha\,
\mathcal S(\ket\psi)
+
\sin\alpha\,
g(R_{\pi/2})
\mathcal S(\ket\psi)
\end{split}
\end{equation}
for an arbitrary vector $\ket\psi\in\mathbb C^d$ and $\alpha\in\mathbb R$. Since $\mathcal S(\ket\psi)$ can be an arbitrary vector from $\mathcal H$ (condition (iv)), we have
\begin{equation}
\label{eq:gRalpha}
g(R_\alpha) = \cos\alpha\,
\mathds{1}
+
\sin\alpha\,
g(R_{\pi/2}).       
\end{equation}    
Also,
\begin{equation}
-\mathcal S(\ket\psi)
=\mathcal S(i^2\ket\psi)
=g(R_{\pi/2})^2\mathcal S(\ket\psi).
\end{equation}
Again, due to the arbitrariness of $\mathcal S(\ket\psi)$, 
\begin{equation}
\label{eq:repsq}
g(R_{\pi/2})^2=g(R_{\pi})=-\mathds{1}.
\end{equation}
We see that $g(R_{\pi/2})$ plays the role of the imaginary unit.

The irreducible real representations of SO(2) are well-known. Namely, there is one trivial one-dimensional representation $g_0(R_\alpha)=1$ and a family of two-dimensional irreducible representations
\begin{equation}
\label{eq:irrep}
g_n(R_\alpha)=
\begin{pmatrix}
\cos n\alpha & 
-\sin n\alpha \\
\sin n\alpha &
\cos n\alpha
\end{pmatrix},
\end{equation}
for any integer $n$.
A general representation of SO(2) then has the form
\begin{equation}
\label{eq:repgen}
g(R_\alpha)=\bigoplus_{n\in\mathbb Z}g_n(R_\alpha)^{\oplus m_n},
\end{equation}
where $m_n\geq0$. The application of Eq.~(\ref{eq:repsq}) implies that $n$ is odd (i.e., $m_n=0$ for even $n$). Actually, only  $n=\pm1$ do not violate linearity of $\mathcal S$, which can be shown, e.g., by consideration of Eq.~(\ref{eq:gRalpha}) for small $\alpha$, but we will not need this.

Since all irreducible representations with  odd $n$ have dimension two,  $\mathcal H$ has the structure of a direct sum of two-dimensional subspaces. For an odd $n$, $g_n(R_{\pi/2})=\pm \begin{pmatrix}
0 & -1\\ 1 & 0
\end{pmatrix}=\pm J_F$ maps every vector to an orthogonal one in each of these subspaces, hence, 
\begin{equation}
\label{eq:orthorot}
\braket{\widetilde\psi|g(R_{\pi/2})|\widetilde\psi}=0
\end{equation}
for an arbitrary $\ket{\widetilde\psi}\in\mathcal H$. 

Consider an arbitrary orthonormal basis $\{\ket{e_j}\}_{j=1}^d$ of $\mathbb C^d$. Condition (iii) implies that the vectors $\mathcal S(\ket{e_j})$ are also orthonormal. Moreover,
\begin{equation}
\label{eq:orthorotSe}
0=\mathcal S(e^{i\alpha}\ket{e_j})^T
\mathcal S(e^{i\beta}\ket{e_k})
=[g(R_{\alpha})\mathcal S(\ket{e_j})]^T
[g(R_{\beta})
\mathcal S(\ket{e_k})]
\end{equation}
for $j\neq k$ and all real $\alpha$ and $\beta$. In view of Eqs.~\eqref{eq:gRalpha} and \eqref{eq:orthorot}, the sets 
\begin{equation}
\mathcal H_j=\{ag(R_{\alpha})\mathcal S(\ket{e_j})\,|\,a\geq0,\,\alpha\in\mathbb R\}
\end{equation} 
are two-dimensional subspaces of $\mathcal H$. Eq.~(\ref{eq:orthorotSe}) means that these subspaces are orthogonal to each other. Since the range of $\mathcal S$ is the whole $\mathcal H$, the dimensionality of the latter is $2d$. We can denote the elements of an orthonormal basis of $\mathcal H\cong\mathbb R^d\otimes\mathbb R_F^2$ as $\ket{\widetilde e_j,0}\cong\ket{\widetilde e_j}\otimes\ket0_F$ and $\ket{\widetilde e_j,1}\cong\ket{\widetilde e_j}\otimes\ket1_F$, where $j$ enumerates the subspaces $\mathcal H_j$ and $0$ and $1$ enumerate the orthonormal basis elements in these subspaces. Namely, we put by definition
\begin{equation}
\mathcal S(\ket{e_j})=\ket{\widetilde e_j}\otimes\ket0_F
\end{equation}
and
\begin{equation}
g(R_{\pi/2})\ket{\widetilde e_j}\otimes\ket0
=\ket{\widetilde e_j}\otimes\ket1_F.
\end{equation}
Thus,
\begin{equation}
\mathcal S(i\ket{e_j})=g(R_{\pi/2})\ket{\widetilde e_j}\otimes\ket0_F
=\ket{\widetilde e_j}\otimes\ket1_F
\end{equation}
and Eq.~(\ref{eq:Sdetailed}) follows from the decomposition $c_j={\rm Re}(c_j)+i\,{\rm Im}(c_j)$.
\end{proof}

\begin{proposition}
\label{prop:uniqueops}
Consider a map $\mathcal{T}$ from the space of linear operators on $\mathbb{C}^d$, i.e., $\mathcal{L}(\mathbb{C}^d)$, to the space of linear operators on a real Hilbert space $\mathcal{H}$, i.e., $\mathcal{L}(\mathcal{H})$, satisfying the following properties:
\begin{enumerate}[(i)]
\item $\mathcal{T}$ is real-linear, i.e.,
\begin{equation}
\mathcal{T}(aA+bB)=a\mathcal{T}(A)+b\mathcal{T}(B),
\end{equation}
for arbitrary $A,B\in\mathcal{L}(\mathcal{H})$ and $a,b\in\mathbb{R}$.
\item $\mathcal{T}(A)\mathcal{S}(\ket{\psi})=\mathcal{S}(A\ket{\psi})$ for $\ket{\psi}\in\mathbb{C}^d$ and $A\in\mathcal{L}(\mathbb{C}^d)$, were $\mathcal{S}$ is given in Eq.~\eqref{eq:S} (Eq.~\eqref{eq:map1partystates} in the main text) and satisfies the conditions of Proposition~\ref{prop:uniquestates}.
\item $[\mathcal{T}(A),R_F(\alpha)]=0$ for an arbitrary flag rotation $R_F(\alpha)$ (see Eq. \eqref{eq:map1phasestates}).
\end{enumerate}
Then it follows that $\mathcal{T}(A)=\mathrm{Re}(A)\otimes I_F+\mathrm{Im}(A)\otimes J_F$.
\end{proposition}

Note that $R_\alpha$ in Proposition~\ref{prop:uniquestates} denotes an abstract element of SO(2) and $R_F(\alpha)$ denotes a rotation operator in $\mathbb R^2_F$ defined in \eqref{eq:map1phasestates}, i.e., a representation of $R_\alpha$.

\begin{proof}
From conditions (i) and (ii) we can expand
\begin{equation}
\begin{aligned}
\mathcal{T}(A)\mathcal{S}(\ket{\psi})&=\mathcal{T}(\mathrm{Re}(A)+i\mathrm{Im}(A))\mathcal{S}(\ket{\psi})\overset{(i)}{=}\mathcal{T}(\mathrm{Re}(A))\mathcal{S}(\ket{\psi})+\mathcal{T}(i\mathrm{Im}(A))\mathcal{S}(\ket{\psi})\\
&\overset{(ii)}{=}\mathcal{T}(\mathrm{Re}(A))\mathcal{S}(\ket{\psi})+\mathcal{S}(i\mathrm{Im}(A)\ket{\psi})=\mathcal{T}(\mathrm{Re}(A))\mathcal{S}(\ket{\psi})+R_F\Bigl(\frac{\pi}{2}\Bigr)\mathcal{S}(\mathrm{Im}(A)\ket{\psi})\\
&\overset{(ii)}{=}\Bigl(\mathcal{T}(\mathrm{Re}(A))+R_F\Bigl(\frac{\pi}{2}\Bigr)\mathcal{T}(\mathrm{Im}(A))\Bigr)\mathcal{S}(\ket{\psi}).
\end{aligned}
\end{equation}
Since the range of $\mathcal{S}$ is the whole $\mathcal{H}$, we can write the following operator equality
\begin{equation}
\label{eq:prop2eq1}
\mathcal{T}(A)=\mathcal{T}(\mathrm{Re}(A))+R_F\Bigl(\frac{\pi}{2}\Bigr)\mathcal{T}(\mathrm{Im}(A)).
\end{equation}

Straightforward calculation shows that property (iii) implies the general form $\mathcal{T}(A)=\mathcal{T}'(A)\otimes I_F+\mathcal{T}''(A)\otimes J_F$ where $\mathcal{T}'(A)$ and $\mathcal{T}''(A)$ are operators acting on $\mathbb R^d$. Now let us consider a real operator $B$ and a vector $\ket{\varphi}$ which we can choose to be real. We expand
\begin{align}
\label{eq:TASpsi}
\mathcal{T}(B)\mathcal{S}(\ket{\varphi})&=(\mathcal{T}'(B)\otimes I_F+\mathcal{T}''(B)\otimes J_F)\ket{\varphi}\otimes\ket{0}_F\\ \nonumber
&=\mathcal{T}'(B)\ket{\varphi}\otimes\ket{0}_F+\mathcal{T}''(B)\ket{\varphi}\otimes\ket{1}_F,\\
\mathcal{S}(B\ket{\varphi})&=B\ket{\varphi}\otimes\ket{0}_F.
\end{align}
Since this equations must be equal for all $\ket{\varphi}$, it follows immediately that $\mathcal{T}'(B)=B$ and $\mathcal{T}''(B)=0$ and thus $\mathcal{T}(B)=B\otimes I_F$ if $B$ is a real operator. Now substituting it  into Eq. \eqref{eq:prop2eq1} and noting that $\mathrm{Im}(A)$ is a real operator, we can rewrite
\begin{equation}
\mathcal{T}(A)=\mathrm{Re}(A)\otimes I_F + R_F\Bigl(\frac{\pi}{2}\Bigr)\mathrm{Im}(A)\otimes I_F=\mathrm{Re}(A)\otimes I_F +\mathrm{Im}(A)\otimes J_F,
\end{equation}
where we have used that $R_F(\frac{\pi}{2})=J_F$.
\end{proof}

Now let us show that the quotient space construction is essentially unique. Consider the case of a composite (for simplicity -- bipartite) system. Consider two complex Hilbert spaces $\mathbb C^{d_A}$  and $\mathbb C^{d_B}$, the vectors $\ket\psi\in\mathbb C^{d_A}$, $\ket\phi\in\mathbb C^{d_B}$ and the corresponding real vectors $\ket{\widetilde\psi}=\mathcal S(\ket\psi)$ and $\ket{\widetilde\phi}=\mathcal S(\ket\phi)$. From one side, according to the general description of the complex-to-real correspondence applied to the space $\mathbb C^{d_A}\otimes \mathbb C^{d_B}$, the complex vector $\ket\psi\otimes\ket\phi$ is mapped into the real vector
\begin{equation}
\mathcal S(\ket\psi\otimes\ket\phi)=
{\rm Re}(\ket\psi\otimes\ket\phi)\otimes\ket 0_F
+
{\rm Im}(\ket\psi\otimes\ket\phi)\otimes\ket 1_F
\in\mathbb R^{d_A+d_B}\otimes\mathbb R^2_F.
\end{equation}
From the other side, consider the quotient space construction. Let us introduce notations 
\begin{equation}
\label{eq:otimesFH}
\mathbb R^{2d_A}
\otimes_F
\mathbb R^{2d_B}
=
\mathbb R^{2d_A}
\otimes
\mathbb R^{2d_B}/{\rm ker}[(\mathcal S^{-1})^{\otimes 2}],
\end{equation}
with the quotient space described in Eq. \eqref{eq:FlagEquivClasses} (and in the text above it) and
\begin{equation}
\label{eq:otimesF}
\ket{\widetilde\psi}\otimes_F\ket{\widetilde\phi}=
[\ket{\widetilde\psi}\otimes\ket{\widetilde\phi}].
\end{equation}
Here $[\cdot]$ denotes the equivalence class of a vector. Namely,
\begin{equation}
\label{eq:otimesFpsiphiKett}
\ket{\widetilde\psi}\otimes_F\ket{\widetilde\phi}
=
(\ket{\widetilde\psi'}\otimes
\ket{\widetilde\phi'}
-\ket{\widetilde\psi''}
\otimes
\ket{\widetilde\phi''})
\otimes
\kett{0}_F
+
(\ket{\widetilde\psi'}\otimes
\ket{\widetilde\phi''}
+\ket{\widetilde\psi''}
\otimes
\ket{\widetilde\phi'})
\otimes
\kett{1}_F
\end{equation}
for $\ket{\widetilde\psi}=\ket{\widetilde\psi'}\otimes\ket0+\ket{\widetilde\psi''}\otimes\ket1$ and $\ket{\widetilde\phi}=\ket{\widetilde\phi'}\otimes\ket0+\ket{\widetilde\phi''}\otimes\ket1$. The notation $\otimes_F$ (``flag tensor product'') is inspired by the symmetrized and antisymmetrized tensor product $\otimes_{\rm s}$ and $\otimes_{\rm a}$ used for indistinguishable particles: a tensor product followed by a kind of symmetrization operation. The replacement of, e.g., $\ket{0}_{F_a}\otimes\ket{0}_{F_b}$ by $\kett{0}_F$, or (see the main text and the next section) $\ket{\psi^{(2)}_{\rm even}}$ is also a kind of symmetrization.

Let us describe notation in Eq.~(\ref{eq:otimesFpsiphiKett}). In the main text, we defined $\kett0_F$ and $\kett1_F$ as equivalence classes only for the flags:

\begin{equation}
\begin{split}
\kett{0}_{F}&=
\ket{\psi_{\mathrm{even}}^{(2)}}_{F_aF_b}
+{\rm span}\{\ket{00}_{F_aF_b}+\ket{11}_{F_aF_b},
\:
\ket{01}_{F_aF_b}-\ket{10}_{F_aF_b}\}
\\
&=
\big\{
\ket{\psi_{\mathrm{even}}^{(2)}}_{F_aF_b}
+a(
\ket{00}_{F_aF_b}+\ket{11}_{F_aF_b}
)
+
b(
\ket{01}_{F_aF_b}-\ket{10}_{F_aF_b}
)
\:\big|\:
a,b\in\mathbb R
\big\}
\end{split}
\end{equation}
and
\begin{equation}
\begin{split}
\kett{1}_{F}&=
\ket{\psi_{\mathrm{odd}}^{(2)}}_{F_aF_b}
+{\rm span}\{\ket{00}_{F_aF_b}+\ket{11}_{F_aF_b},
\:
\ket{01}_{F_aF_b}-\ket{10}_{F_aF_b}\}\\
&=
\big\{
\ket{\psi_{\mathrm{odd}}^{(2)}}_{F_aF_b}
+a(
\ket{00}_{F_aF_b}+
\ket{11}_{F_aF_b}
)
+
b(
\ket{01}_{F_aF_b}-\ket{10}_{F_aF_b}
)
\:\big|\:
a,b\in\mathbb R
\big\}.
\end{split}
\end{equation}
Then, $\ket\chi\otimes\kett0_F$ for $\ket\chi\in\mathbb R^{d_A}\otimes\mathbb R^{d_B}$ denotes the equivalence class
\begin{equation}
\begin{split} \ket\chi\otimes\kett0_F    &=\ket\chi\otimes
\ket{\psi_{\mathrm{even}}^{(2)}}_{F_aF_b}
+{\rm ker}[(\mathcal S^{-1})^{\otimes 2}]
\\
&=
\big\{
\ket\chi\otimes
\ket{\psi_{\mathrm{even}}^{(2)}}_{F_aF_b}
+
\ket\xi\otimes
(\ket{00}_{F_aF_b}+
\ket{11}_{F_aF_b}
)
+
\ket\zeta\otimes
(\ket{01}_{F_aF_b}-\ket{10}_{F_aF_b}
)
\:\big|\:
\ket\xi,\ket\zeta\in\mathbb R^{d_A}\otimes\mathbb R^{d_B}
\big\},
\end{split}    
\end{equation}
and analogously for $\ket\chi\otimes\kett1_F$. Note that $\ket\xi$ and $\ket\zeta$ are not necessarily normalized. The norms of $\ket\chi\otimes\kett0_F$ and $\ket\chi\otimes\kett1_F$ as elements of the quotient space are equal to the norms of the canonical representatives of the corresponding classes $\ket\chi\otimes
\ket{\psi_{\mathrm{even}}^{(2)}}_{F_aF_b}$ and $\ket\chi\otimes
\ket{\psi_{\mathrm{odd}}^{(2)}}_{F_aF_b}$, respectively.

A natural isomorphism between $\mathbb R^{d_A+d_B}\otimes \mathbb R^2_F$ and $\mathbb R^{d_A}\otimes_F\mathbb R^{d_B}$ is defined by identification of the vectors $\ket0_F$ and $\ket1_F$ from $\mathbb R^2_F$ with the flag equivalence classes $\kett0_F$ and $\kett1_F$ in the quotient space construction. Now we are ready to formulate that, up to an isomorphism, the quotient space construction is unique.

\begin{proposition}
\label{prop:uniqueOtimesF}
If $\mathcal S$ is real-linear and is of the form given in Eq.~(\ref{eq:S}), then
\begin{equation}
\mathcal S(\ket\psi\otimes\ket\phi)\cong
\mathcal
S(\ket\psi)\otimes_F \mathcal
S(\ket\phi).
\end{equation}
\end{proposition}

\begin{proof}
We have
\begin{equation}
\label{eq:otimesFderiv}
\begin{split}
\mathcal S(\ket\psi\otimes\ket\phi)&=
{\rm Re}(\ket\psi\otimes\ket\phi)\otimes\ket 0_F
+
{\rm Im}(\ket\psi\otimes\ket\phi)\otimes\ket 1_F
\\
&=\big[{\rm Re}(\ket{\psi})\otimes
{\rm Re}(\ket{\phi})
-{\rm Im}(\ket{\psi})
\otimes
{\rm Im}(\ket{\phi})\big]
\otimes
\ket{0}_F
+\big[{\rm Re}(\ket{\psi})\otimes
{\rm Im}(\ket{\phi})
-{\rm Im}(\ket{\psi})
\otimes
{\rm Re}(\ket{\phi})\big]
\otimes
\ket{1}_F
\\&\cong
\mathcal
S(\ket\psi)\otimes_F \mathcal
S(\ket\phi).
\end{split}
\end{equation}
\end{proof}

Of course, in the proof of Proposition~\ref{prop:uniqueOtimesF}, we have not ``derived'' the quotient space construction from scratch, but only showed an isomorphism when this construction has been already built. But the quotient space construction is natural to express cumbersome expressions like in Eqs.~(\ref{eq:otimesFpsiphiKett}) and (\ref{eq:otimesFderiv}) in a compact form as in Eq.~(\ref{eq:otimesF}).

There are real-number formulations of quantum mechanics in the infinite dimensional case, see, for example, Refs.~\cite{Stueckelberg, volovichkahler}. We believe that our construction can also be extended to infinite dimensions because we do not see any fundamental obstacles for this. Of course, certain attention will have to be paid to the domains of unbounded operators in the real-to-complex map.

\section{State map: General case}
\label{ap:state_map_gen_case}

In the main text and the previous section, we gave a construction of the Hilbert space of an arbitrary bipartite composite system. This is sufficient for the logical construction of the theory because, in the multipartite case, we can apply the bipartite rule sequentially. However, it is interesting to analyse the structure of the Hilbert space of a composite system for the  multipartite case.

For any $d\in\mathbb{N}$, the mapping between complex and real state vectors proposed in \cite{Myrheim} and \cite{McKague2009} reads
\begin{equation}
\begin{aligned}
\mathcal{S}\colon \mathbb{C}^d &\to\mathbb{R}^{2d},\\
\ket{\psi} &\mapsto \mathrm{Re}(\ket{\psi}) \otimes \ket{0}_F + \mathrm{Im}(\ket{\psi}) \otimes \ket{1}_F,
\end{aligned}
\end{equation}
where the real and imaginary parts of the state vectors are defined in a fixed basis. 

Now we extend the isomorphism $\mathcal{S}$ to systems composed of $N$ distinct parties. This extension guarantees that the resulting real and complex quantum theories are isomorphic, ensuring that any experiment expressible within complex quantum mechanics can equally be described within real quantum mechanics. For simplicity of notations, we assume that all subsystems have the same dimensionality $d$, but the analysis is completely the same in the case of different dimensionalities.

A natural first attempt is for $\ket{\psi_1},\ket{\psi_2},\ldots,\ket{\psi_N}\in\mathbb{C}^d$ to consider $\mathcal{S}^{\otimes N}\colon\ket{\psi_1}\otimes\cdots\otimes\ket{\psi_N}\mapsto \mathcal{S}(\ket{\psi_1})\otimes\cdots\otimes \mathcal{S}(\ket{\psi_N})$.
However, as discussed in the main text, this expression is not a well defined map. Thus one has to follow a different approach.
Note that $\mathcal{S}$ is invertible and the extension of $\mathcal{S}^{-1}$ to the tensor product of real Hilbert spaces (denoted by $(\mathcal{S}^{-1})^{\otimes N}$) is well defined. However, two elements in the domain of $(\mathcal{S}^{-1})^{\otimes N}$ have the same image and so their difference is mapped to the zero element, i.e., $\text{ker}[(\mathcal{S}^{-1})^{\otimes N}]$ is non-trivial.
This implies that $(\mathcal{S}^{-1})^{\otimes N}$ is not invertible (see Figure~\ref{fig:balloons}).
To overcome this issue we consider the quotient space of the domain by the kernel.
In this space the map is invertible, reading
\begin{equation}
\label{eq:Q}
\mathcal{Q}\colon\mathbb{C}^d\otimes\cdots\otimes\mathbb{C}^d\to\mathbb{R}^{2d}\otimes\cdots\otimes\mathbb{R}^{2d}/\text{ker}[(\mathcal{S}^{-1})^{\otimes N}].
\end{equation}
Every element in the quotient space is an equivalence class defined by the equivalence relation: $\ket{v}$ and $\ket{w}$ are equivalent (denoted as $\ket{v}\sim \ket{w}$) if $\ket{v}-\ket{w} \in \text{ker}[(\mathcal{S}^{-1})^{\otimes N}]$.

With this in mind we can distinguish $2d^N$ linearly independent equivalence classes, namely $[\ket{r_i}\otimes\ket{00\ldots0}_F]$ and $[\ket{r_i}\otimes\ket{10\ldots0}_F]$,
for $\{\ket{r_i}\}$ a basis of $\mathbb{R}^{d^N}$. We use the subscript ``$F$'' to denote to the flag states of every party for clearer notation.
And with these definitions,
\begin{equation}
(\mathbb{R}^{d})^{\otimes N}\otimes(\mathbb{R}^{2})^{\otimes N}/\text{ker}[(\mathcal{S}^{-1})^{\otimes N}]= \text{span}\Biggl(\bigcup_{i=1}^{d^N}\Big\{\big[\ket{r_i}\otimes\ket{00\ldots0}_F\big],		\big[\ket{r_i}\otimes\ket{10\ldots0}_F\big]\Big\}\Biggr).
\end{equation}

Now, for any $\ket{v}\in(\mathbb{R}^{d})^{\otimes N}\otimes(\mathbb{R}^{2})^{\otimes N}$, the canonical form of a representative of the equivalence class $[\ket{v}]$ is given by $\ket{v}^\perp =P^\perp\ket{v},$
where in our case $P^\perp$ is the projection operator onto $\text{ker}[(\mathcal{S}^{-1})^{\otimes N}]^\perp$. The element $\ket{v}^\perp$ is thus called the canonical representative of the equivalence class. It is unique and independent of the choice of $\ket{v}$ from the equivalence class and therefore identifies it. Thus, we define the projection
\begin{align}
\label{eq:P}
\mathcal{P}\colon (\mathbb{R}^{d})^{\otimes N}\otimes(\mathbb{R}^{d})^{\otimes N}/\text{ker}[(\mathcal{S}^{-1})^{\otimes N}]&\to \text{ker}[(\mathcal{S}^{-1})^{\otimes N}]^\perp
\subset (\mathbb{R}^{d})^{\otimes N},\\
[\ket{v}]&\mapsto P^\perp\ket{v}\nonumber,
\end{align}
the explicit form of the projector is 
\begin{equation}
\label{eq:Pperp}
P^\perp~=~\mathds{1}~\otimes~(\ket{\psi_{\mathrm{even}}^{(N)}}\bra{\psi_{\mathrm{even}}^{(N)}}_F+\ket{\psi_{\mathrm{odd}}^{(N)}}\bra{\psi_{\mathrm{odd}}^{(N)}}_F),
\end{equation}
where the subscript $F$ indicates that the projectors act on the flag subspace, and 
\begin{equation}
\begin{aligned}
\label{evenoddstates}	\ket{\psi_{\mathrm{even}}^{(N)}}&=\frac{1}{\sqrt{2^{N-1}}}\sum_{\substack{k\in\{0,1\}^N\\|k|_{\rm H} \text{ is even}}}(-1)^{\frac{|k|_{\rm H}}{2}}\ket{k},\\	\ket{\psi_{\mathrm{odd}}^{(N)}}&=\frac{1}{\sqrt{2^{N-1}}}\sum_{\substack{k\in\{0,1\}^N\\|k|_{\rm H} \text{ is odd}}}(-1)^{\frac{|k|_{\rm H}-1}{2}}\ket{k},
\end{aligned}
\end{equation}
and $|k|_{\rm H}$ is the sum of the binary digits $k_1,k_2,\ldots,k_N$ of $k$, also called the Hamming weight. We have included a normalization constant because these vectors will be used in the representation of quantum states that must satisfy postulate (P1). 

We now present how our construction extends states and operators from complex quantum mechanics to our proposed theory. 
A generic state in complex quantum mechanics $\ket{\psi}\in\mathbb{C}^d\otimes\cdots\otimes\mathbb{C}^d$ gets mapped into a real state in terms of canonical representatives of equivalence classes as
\begin{align}
\nonumber	
\mathcal{R}\colon\ket{\psi}\mapsto \mathcal{R}(\ket{\psi})=&\frac{\mathcal{P}\circ \mathcal{Q}(\ket{\psi})}{\|\mathcal{P}\circ \mathcal{Q}(\ket{\psi})\|}\\
\nonumber
=&\sqrt{2^{N-1}}\Bigl(P^\perp\big([\mathrm{Re}(\ket{\psi})\otimes\ket{00\ldots0}]\big)
+P^\perp\big([\mathrm{Im}(\ket{\psi})\otimes\ket{10\ldots0}]\big)\Bigr)\\
=&\mathrm{Re}(\ket{\psi})\otimes \ket{\psi_{\mathrm{even}}^{(N)}}_F + \mathrm{Im}(\ket{\psi})\otimes \ket{\psi_{\mathrm{odd}}^{(N)}}_F.
\label{eq:mapstates}            
\end{align}

\begin{remark}
\label{rem:tworepresentations}
The map of Eq.~(\ref{eq:map1partystates}) from the main text (applied to a composite system) and the one given in Eq.~(\ref{eq:mapstates}) are different representations of the same map. On the one hand, the flag Hilbert space is isomorphic to $\mathbb R^2$. The notations $\ket{0}_F$ and $\ket{1}_F$ from the main text denote the basis vectors from $\mathbb R^2$. On the other hand, the Hilbert space of a composite system is a quotient space of the tensor product of Hilbert spaces with respect to a certain equivalence relation. Then ${\rm Re}(\ket\psi)\otimes\kett{0}_F$ and ${\rm Im}(\ket\psi)\otimes\kett{1}_F$ are elements of this quotient space and, formally, 
\begin{equation}
\begin{aligned}
\mathcal S\colon (\mathbb C^d)^{\otimes N}&\to
(\mathbb R^d)^{\otimes N}\otimes \mathbb R^2
\cong
(\mathbb R^d\otimes\mathbb R^2)^{\otimes N}/{\rm ker}[(\mathcal S^{-1})^{\otimes N}]\\
\ket{\psi}&\mapsto\mathrm{Re}(\ket{\psi})\otimes\kett{0}_F+\mathrm{Im}(\ket{\psi})\otimes\kett{1}_F,
\end{aligned}
\end{equation}
where
$\cong$ denotes the Hilbert space isomorphism. Instead of equivalence classes ${\rm Re}(\ket\psi)\otimes\kett{0}_F$ and ${\rm Im}(\ket\psi)\otimes\kett{1}_F$, we can consider their canonical representatives ${\rm Re}(\ket\psi)\otimes\ket{\psi_{\mathrm{even}}^{(N)}}_F$ and ${\rm Im}(\ket\psi)\otimes\ket{\psi_{\mathrm{odd}}^{(N)}}_F$, which belong to  the tensor product $(\mathbb R^2)^{\otimes N}$. Thus, 
\begin{equation}
\begin{aligned}
\mathcal R\colon (\mathbb C^d)^{\otimes N}&\to
(\mathbb R^d\otimes\mathbb R^2)^{\otimes N}
\cong
(\mathbb R^d)^{\otimes N}
\otimes
\mathbb R^{2N}\\
\ket{\psi}&\mapsto\mathrm{Re}(\ket{\psi})\otimes\ket{\psi_{\mathrm{even}}^{(N)}}_F+\mathrm{Im}(\ket{\psi})\otimes\ket{\psi_{\mathrm{odd}}^{(N)}}_F
.
\end{aligned}
\end{equation}
These two approaches (working with the equivalence classes or their canonical representations) are equivalent. In the following, we will not distinguish these two cases and use the same notations for both. E.g., in the main text, we defined the map $\mathcal T\colon\mathcal L(\mathbb C^d)\to\mathcal L(\mathbb R^d\otimes\mathbb R^2)$ (with $\mathcal L$ denoting the space of linear operators in the corresponding Hilbert space). In the next subsection, we will use the same notation for a map from $\mathcal L\big((\mathbb C^d)^{\otimes N}\big)$ to $\mathcal L\big((\mathbb R^d\otimes\mathbb R^2)^{\otimes N}\big)$ (instead of $\mathcal L\big((\mathbb R^d)^{\otimes N}\otimes\mathbb R^2\big)$. Also we will not distinguish between equivalence classes and their canonical representations. But for vectors, the notations $\mathcal R$ and $\mathcal S$ are different because the map $\mathcal S$ is used in the definition of $\mathcal R$, see Eqs.~(\ref{eq:Q}), (\ref{eq:P}), and (\ref{eq:mapstates}).

\end{remark}

\begin{remark}
Let us also note that the vectors $\ket{\psi^{(N)}_{\rm even}}$ and $\ket{\psi^{(N)}_{\rm odd}}$ are not preserved under the tensor product. For example, $\ket{\psi^{(N)}_{\rm even}}
\otimes
\ket{\psi^{(M)}_{\rm even}}$ is not a linear combination of $\ket{\psi^{(N+M)}_{\rm even}}$ and $\ket{\psi^{(N+M)}_{\rm odd}}$. Instead of this, we have the following relations:
\begin{equation}
\label{eq:psievenoddNM}
\begin{split}
\ket{\psi^{(N+M)}_{\rm even}}
&=
\frac1{\sqrt2}
(\ket{\psi^{(N)}_{\rm even}}
\otimes
\ket{\psi^{(M)}_{\rm even}}
-
\ket{\psi^{(N)}_{\rm odd}}
\otimes
\ket{\psi^{(M)}_{\rm odd}}),\\
\ket{\psi^{(N+M)}_{\rm odd}}
&=
\frac1{\sqrt2}
(\ket{\psi^{(N)}_{\rm even}}
\otimes
\ket{\psi^{(M)}_{\rm odd}}
+
\ket{\psi^{(N)}_{\rm odd}}
\otimes
\ket{\psi^{(M)}_{\rm even}}).
\end{split}
\end{equation}
In this sense, the composition of two subsystems is not a ``local'' operation from the viewpoint of the tensor product $(\mathbb R^d\otimes\mathbb R^2_F)^{\otimes(N+M)}$.  But the composition is local from the viewpoint of the quotient space, where we do not distinguish certain vectors from $(\mathbb R^d\otimes\mathbb R^2_F)^{\otimes(N+M)}$ and consider equivalence classes. If $[\ket\psi]$ denotes the equivalence class of the vector $\ket\psi$, then, as we discussed in the main text,
\begin{equation}
\begin{split}
[\ket{\psi^{(N)}_{\rm even}}
\otimes
\ket{\psi^{(M)}_{\rm even}}]
&=[-\ket{\psi^{(N)}_{\rm odd}}
\otimes
\ket{\psi^{(M)}_{\rm odd}}]
= [\ket{\psi^{(N+M)}_{\rm even}}]
=\kett{0}_F,\\        
[\ket{\psi^{(N)}_{\rm even}}
\otimes
\ket{\psi^{(M)}_{\rm odd}}]
&=[\ket{\psi^{(N)}_{\rm odd}}
\otimes
\ket{\psi^{(M)}_{\rm even}}]
= [\ket{\psi^{(N+M)}_{\rm odd}}]        =\kett{1}_F.
\end{split}    
\end{equation}
This underlines the importance of the quotient space construction, which is central for this paper, while the even and odd vectors themselves were known before \cite{McKague2009}. 
If we do not distinguish between the equivalence classes and their canonical representatives, Eq.~(\ref{eq:otimesFpsiphiKett}) can be rewritten as
\begin{equation}
\label{eq:otimesFpsiphi}
\ket{\widetilde\psi}\otimes_F\ket{\widetilde\phi}
=
(\ket{\widetilde\psi'}\otimes
\ket{\widetilde\phi'}
-\ket{\widetilde\psi''}
\otimes
\ket{\widetilde\phi''})
\otimes
\ket{\psi^{(2)}_{\rm even}}
+
(\ket{\widetilde\psi'}\otimes
\ket{\widetilde\phi''}
+\ket{\widetilde\psi''}
\otimes
\ket{\widetilde\phi'})
\otimes
\ket{\psi^{(2)}_{\rm odd}}.
\end{equation}
\end{remark}

\begin{remark}
A product state in complex quantum mechanics can be mapped to a non-product state (not separable into tensor product factors) in the real-valued formalism. It is therefore natural to ask whether the converse is true: Are there cases where an entangled state in complex quantum mechanics gets mapped to a product state in the real version? The answer is no. 
The canonical representative states of the flag $\ket{\psi^{(N)}_{\rm even}}$ and $\ket{\psi^{(N)}_{\rm odd}}$ are always of non-product form.
\end{remark}

\section{Alternative derivation of the complex-to-real map for operators}
\label{ap:map_op_gen}

In this section, we give an alternative derivation of the map $\mathcal T$ given in Eq.~(\ref{eq:RPi}), which maps complex operators to their real counterparts based on the canonical representatives approach (see Remark~\ref{rem:tworepresentations}).

Let $A=\sum_j \ket{k_j}\bra{e_j}$ with $\ket{k_j},\ket{e_j}\in (\mathbb{C}^d)^{\otimes N}$ be a general linear operator acting on $N$ qudits of dimension $d$.
Because we already know how $\mathcal{R}$ acts on vectors, it is convenient to apply it to a vectorization of $A$, $\sum_j\ket{k_j}\otimes \ket{e_j}$, and obtain $\mathcal{R}(\sum_j \ket{k_j}\otimes\ket{e_j})$. 
Now if $\mathcal{R}$ would extend to tensor products, this would give $\sum_j \mathcal R(\ket{k_j})\otimes \mathcal R(\ket{e_j})$ and undoing the vectorization gives $\sum_j \mathcal R(\ket{k_j}) \mathcal R(\ket{e_j})^T$. 
One might be tempted to view this as the natural extension of $\mathcal{R}$ to operators.
However, similar to what we saw in the main text for states, the expression $\sum_i\ket{k_j}\bra{e_j}\mapsto \sum_j \mathcal{R}(\ket{k_j}) \mathcal{R}(\ket{e_j})^T$ is not a well-defined map. 
This is due to the difference in complex and real linear structure between the domain and the image. For example, despite $\ket{\psi}\otimes(\ket{\phi})^\dagger$ and $i\ket{\psi}\otimes(i\ket{\phi})^\dagger$ being  the same element of the domain, they do not have the same image, see also Figure~\ref{fig:balloons}.
However, the inverse map 
\begin{equation}
\label{eq:Rinv}
\mathcal{R}^{-1}(\ket{\widetilde\phi}) = \langle\psi_{\mathrm{even}}^{(N)}|\widetilde\phi\rangle+ i\langle\psi_{\mathrm{odd}}^{(N)}|\widetilde\phi\rangle,
\end{equation}
does extend to tensor products: $(\mathcal{R}^{-1})^{\otimes 2}(\ket{\widetilde\phi}\otimes \ket{\widetilde\chi})=\mathcal{R}^{-1}(\ket{\widetilde\phi})\otimes \mathcal{R}^{-1}(\ket{\widetilde\chi})$. Note that the right-hand side in Eq.~(\ref{eq:Rinv}) is a vector since   the scalar product is taken only for the flag part. Hence, we can use the vectorization described above to define
\begin{equation}
{\mathcal{T}}^{-1}_{\rm L}\left(\sum_j\ket{\widetilde\phi_j}\bra{\widetilde\chi_j}\right) = \frac{1}{2}\sum_j \mathcal{R}^{-1}(\ket{\widetilde\phi_j})\mathcal{R}^{-1}(\ket{\widetilde\chi_j})^T,  \label{eq:Rinverse}
\end{equation}
where we added a factor of $\frac{1}{2}$ such that, as will become clear later, expectation values are reproduced.
This map is not invertible due to the non-triviality of the kernel: Any element of the form
\begin{equation}
\label{eq:quo-space-ops}
\sum_j \left(\mathcal{R}(\ket{k_j})\mathcal{R}(\ket{e_j})^T-\mathcal{R}(e^{i\alpha_j}\ket{k_j})\mathcal{R}(e^{i\alpha_j}\ket{e_j})^T\right) 
\end{equation}
with real phases $\alpha_j$ is mapped to zero. However, as we are going to show, $\mathcal T^{-1}_{\rm L}$ is the left inverse of an operator $\mathcal T$ to be defined.
In order to define it, we consider the quotient by the kernel.
Elements of the resulting equivalence classes differ only in the phases.
As a canonical representative for the equivalence classes we therefore choose
\begin{equation}
\label{eq:canrepop}
\sum_j \mathcal{R}(\ket{k_j})\mathcal{R}(\ket{e_j})^T+\mathcal{R}(i\ket{k_j})\mathcal{R}(i\ket{e_j})^T\in\text{ker}[(\mathcal{R}^{-1})^{\otimes 2}]^\perp,
\end{equation}
which, as can be seen from direct calculation, is independent of any additional phases.
Inserting Eq.~\eqref{eq:mapstates} into Eq.~\eqref{eq:canrepop} gives 
\begin{align}
\mathcal{T}(A):=&\sum_j \mathcal{R}(\ket{k_j})\mathcal{R}(\ket{e_j})^T+\mathcal{R}(i\ket{k_j})\mathcal{R}(i\ket{e_j})^T
\nonumber\\
=&\sum_j\Bigl(\mathrm{Re}(\ket{k_j})\otimes\ket{\psi_{\mathrm{even}}^{(N)}}+\mathrm{Im}(\ket{k_j})\otimes\ket{\psi_{\mathrm{odd}}^{(N)}}\Bigr)
\Bigl(\mathrm{Re}(\ket{e_j})^T\otimes\bra{\psi_{\mathrm{even}}^{(N)}}+\mathrm{Im}(\ket{e_j})^T\otimes\bra{\psi_{\mathrm{odd}}^{(N)}}\Bigr)
\nonumber\\
&+\sum_j\Bigl(-\mathrm{Im}(\ket{k_j})\otimes\ket{\psi_{\mathrm{even}}^{(N)}}+\mathrm{Re}(\ket{k_j})\otimes\ket{\psi_{\mathrm{odd}}^{(N)}}\Bigr)
\Bigl(-\mathrm{Im}(\ket{e_j})^T\otimes\bra{\psi_{\mathrm{even}}^{(N)}}+\mathrm{Re}(\ket{e_j})^T\otimes\bra{\psi_{\mathrm{odd}}^{(N)}}\Bigr)
\nonumber\\
=&\sum_i\mathrm{Re}(\ket{k_i}\bra{e_i})\otimes I^{(N)}+\mathrm{Im}(\ket{k_i}\bra{e_i})\otimes J^{(N)}
\nonumber\\
=&\mathrm{Re}(A)\otimes I^{(N)}+\mathrm{Im}(A)\otimes J^{(N)}, 
\label{eq:Roperator}
\end{align}
where we introduced the notations
\begin{equation}
\label{eq:IJN}
I^{(N)}=\ket{\psi^{(N)}_{\rm even}}\bra{\psi^{(N)}_{\rm even}}+\ket{\psi^{(N)}_{\rm odd}}\bra{\psi^{(N)}_{\rm odd}},\qquad
J^{(N)}=\ket{\psi^{(N)}_{\rm odd}}\bra{\psi^{(N)}_{\rm even}}-\ket{\psi^{(N)}_{\rm even}}\bra{\psi^{(N)}_{\rm odd}}.
\end{equation}
We have obtained Eq.~\eqref{eq:RPi} in the main text since $I^{(N)}$ and $J^{(N)}$ are actually realizations of the operators $I_F$ and $J_F$ for the subspace spanned by $\ket{\psi^{(N)}_{\rm even}}$ and $\ket{\psi^{(N)}_{\rm odd}}$. Formally (see Remark~\ref{rem:tworepresentations}), $I_F$ and $J_F$ act in $\mathbb R^2$, while $I^{(N)}$ and $J^{(N)}$ act in $(\mathbb R^2)^{\otimes N}$, but we only interested in their action in the subspace of the canonical representatives of the equivalence classes (i.e., spanned by $\ket{\psi^{(N)}_{\rm even}}$ and $\ket{\psi^{(N)}_{\rm odd}}$), but, as we said in Remark~\ref{rem:tworepresentations}, we will not distinguish between these two approaches.

We verify that ${\mathcal{T}}^{-1}_{\rm L}$ is indeed the left inverse of $\mathcal{T}$: Applying Eq.~\eqref{eq:Rinverse} to Eq.~\eqref{eq:Roperator} gives
\begin{equation}
\begin{aligned}
{\mathcal{T}}^{-1}_{\rm L} ({\mathcal{T}}(A)) =& {\mathcal{T}}^{-1}_{\rm L} \left({\mathcal{T}}\left(\sum_j \ket{k_j}\bra{e_j}\right)\right)\\
=& {\mathcal{T}}^{-1}_{\rm L}\left(\sum_j\mathcal{R}(\ket{k_j})\mathcal{R}(\ket{e_j})^T\right)+{\mathcal{T}}^{-1}_{\rm L}\left(\sum_j\mathcal{R}(i\ket{k_j})\mathcal{R}(i\ket{e_j})^T\right)\\
=&  \frac{1}{2} \sum_j\ket{k_j}\bra{e_j} + \frac{1}{2}\sum_j i \ket{k_j}(-i\bra{e_j})\\
=& A.
\end{aligned}
\end{equation}
Note, that the choice of canonical representative in Eq.~\eqref{eq:canrepop} is not unique. Namely, multiplication by any real coefficient would also give an element of $\text{ker}[(\mathcal{R}^{-1})^{\otimes 2}]^\perp$ since it would still be orthogonal to any element of the form of Eq.~\eqref{eq:quo-space-ops}. However, the canonical representative chosen is the only one that leads to an algebraic homomorphism, i.e., ${\mathcal{T}}(A_1A_2)={\mathcal{T}}(A_1){\mathcal{T}}(A_2)$.

Using the definition of Eq.~(\ref{eq:IJN}) of $I^{(N)}$ and $J^{(N)}$ and Eqs.~\eqref{eq:psievenoddNM}, the following relations hold:
\begin{equation}
\label{eq:IJN+M}
I^{(N+M)}=\frac12
(I^{(N)}\otimes I^{(M)}
-J^{(N)}\otimes J^{(M)}),
\qquad
J^{(N+M)}=\frac12
(I^{(N)}\otimes J^{(M)}
+J^{(N)}\otimes I^{(M)}),
\end{equation}
However, if we consider the action of these operators on the quotient space or, equivalently, the subspace of the canonical representatives, all of them are equivalent, i.e., act as $I_{F}$ and $J_{F}$ on the subspace of spanned by $\ket{\psi^{(N)}_{\rm even}}$ and $\ket{\psi^{(N)}_{\rm odd}}$: 
\begin{equation}
I^{(N+M)}\sim
I^{(N)}\otimes I^{(M)}
\sim-J^{(N)}\otimes J^{(M)}\sim I_{F},
\qquad 
J^{(N+M)}\sim
I^{(N)}\otimes J^{(M)}
\sim J^{(N)}\otimes I^{(M)}\sim J_{F}.
\end{equation}
In particular, if we consider, e.g., $M=1$, we arrive at the conclusion that $J^{(N)}$ (considered as the operator in the quotient space) can be implemented locally using $J^{(1)}=XZ$ acting on any flag qubit, where $X$ and $Z$ are the Pauli matrices. This was noticed in Ref.~\cite{McKague2009}. For two real operators $A$ and $B$ as given in Eq.~(\ref{eq:RPi}) of the main text, we will use the notation $A\otimes_FB$ [cf. Eqs.~\eqref{eq:otimesF}, \eqref{eq:otimesFpsiphiKett} and \eqref{eq:otimesFpsiphi}] to underline that we are interested in the action of the operator $A\otimes B$ only on the subspace spanned by the canonical representatives. One can think that
\begin{equation}
\label{eq:AotimesFB}
A\otimes_F B=
(A'\otimes B'-A''\otimes B'')\otimes I^{(2)}
+
(A'\otimes B''+A''\otimes B')\otimes J^{(2)}
\end{equation}
for $A=A'\otimes I^{(1)}+A''\otimes J^{(1)}$ and $B=B'\otimes I^{(1)}+B''\otimes J^{(1)}$, but if we act only on vectors from the subspace of the canonical representatives $\ket{\widetilde\psi}\otimes_F \ket{\widetilde\phi}$, then we can simply take $A\otimes B$ since its action on vector of this form is the same.

\section{Further properties of the construction}
\label{ap:cons_construction}
In this section we will prove some important properties of the described construction.

\begin{proposition} 
\label{prop:herm2sym}
${\mathcal{T}}$ maps Hermitian operators to real symmetric operators.
\end{proposition}
\begin{proof}
Let $H$ be a Hermitian operator. Then
\begin{equation}
\begin{aligned}
{\mathcal{T}}(H)^T=&\mathrm{Re}(H)^T\otimes I^{(N)}+\mathrm{Im}(H)^T\otimes(J^{(N)} ) ^T=\mathrm{Re}(H)\otimes I^{(N)}-\mathrm{Im}(H)\otimes(J^{(N)})^T\\
=&\mathrm{Re}(H)\otimes I^{(N)}+\mathrm{Im}(H)\otimes J^{(N)}=\mathcal{T}(H),
\end{aligned}
\end{equation}
where we have used  $(J)^T=-J$.
\end{proof}

\begin{proposition} ${\mathcal{T}}$ maps unitary operators to orthogonal operators. 
\end{proposition}
\begin{proof}
Let $U$ be unitary. Then
\begin{equation}
\begin{aligned}
\mathds{1}&=UU^\dagger= [\mathrm{Re}(U)+i\mathrm{Im}(U)][\mathrm{Re}(U)^T-i\mathrm{Im}(U)^T]\\
&=\mathrm{Re}(U)\mathrm{Re}(U)^T+\mathrm{Im}(U)\mathrm{Im}(U)^T+i[-\mathrm{Re}(U)\mathrm{Im}(U)^T+\mathrm{Im}(U)\mathrm{Re}(U)^T],
\end{aligned}
\end{equation}
By comparison, then
\begin{align}
\label{eq16}
\mathrm{Re}(U)\mathrm{Re}(U)^T+\mathrm{Im}(U)\mathrm{Im}(U)^T&=\mathds{1},\\
\label{eq17}
-\mathrm{Re}(U)\mathrm{Im}(U)^T+\mathrm{Im}(U)\mathrm{Re}(U)^T&=0.
\end{align}
Now consider
\begin{equation}
\begin{aligned}
{\mathcal{T}}(U){\mathcal{T}}(U)^T&=[\mathrm{Re}(U)\otimes I^{(N)}+\mathrm{Im}(U)\otimes J^{(N)}][\mathrm{Re}(U)^T\otimes I^{(N)}-\mathrm{Im}(U)^T\otimes J^{(N)}]\\
&=
[\mathrm{Re}(U)\mathrm{Re}(U)^T+\mathrm{Im}(U)\mathrm{Im}(U)^T] \otimes I^{(N)}
+[-\mathrm{Re}(U)\mathrm{Im}(U)^T+\mathrm{Im}(U)\mathrm{Re}(U)^T]\otimes J^{(N)}\\&=
\mathds{1}\otimes I^{(N)},
\end{aligned}
\end{equation}
where we have used $(J^{(N)})^2=-I^{(N)}$.
\end{proof}

Let us now discuss scalar product and distinguishability of states.
In the following we will often use $\mathrm{Re}(ab) = \mathrm{Re}(a)\mathrm{Re}(b) - \mathrm{Im}(a)\mathrm{Im}(b)$ and $\mathrm{Im}(ab) = \mathrm{Re}(a)\mathrm{Im}(b) + \mathrm{Im}(a)\mathrm{Re}(b)$. 
The distinguishability of two states is related to the scalar product of two states: Given a state $\ket{\psi}$ how likely is it to measure $\ket{\phi}$? The probability is given by the Born rule, which in complex quantum mechanics for pure states reads $|\braket{\phi|\psi}|^2$.
Direct calculation shows that
\begin{align}
\nonumber
\mathcal{R}(\ket{\phi})^T \mathcal{R}(\ket{\psi}) &= 
[\mathrm{Re}(\ket{\phi}) \otimes \ket{\psi_\mathrm{even}^{(N)}} + \mathrm{Im}(\ket{\phi}) \otimes \ket{\psi_\mathrm{odd}^{(N)}}]^T [\mathrm{Re}(\ket{\psi}) \otimes \ket{\psi_\mathrm{even}^{(N)}} + \mathrm{Im}(\ket{\psi}] \otimes \ket{\psi_\mathrm{odd}^{(N)}}]  \\
\nonumber
&= [\mathrm{Re}(\bra{\phi}) \otimes \bra{\psi_\mathrm{even}^{(N)}} - \mathrm{Im}(\bra{\phi}) \otimes \bra{\psi_\mathrm{odd}^{(N)}}][\mathrm{Re}(\ket{\psi}) \otimes \ket{\psi_\mathrm{even}^{(N)}} + \mathrm{Im}(\ket{\psi}) \otimes \ket{\psi_\mathrm{odd}^{(N)}}]  \\
&= \mathrm{Re}\braket{\phi|\psi}, 
\label{eq:scalarprod}
\end{align}
i.e., the scalar product is not preserved. However, what is important for the Born rule is preservation of the probabilities of the outcomes and average values, e.g., expressions of the form $\braket{\psi|A|\psi}$, where $A$ is a Hermitian operator. Let us prove this property:
\begin{equation}
\label{eq:expvaluesR}
\mathcal R(\ket\psi)^T \mathcal{T}(H) \mathcal R(\ket\psi)=
\braket{\psi|H|\psi}.
\end{equation}
Let us first prove that 
\begin{equation}
\label{eq:TARpsi}
\mathcal T(A)\mathcal R(\ket\psi)
=
\mathcal R(A\ket\psi)
\end{equation}
for an arbitrary complex linear operator $A$ and complex vector $\ket{\psi}$. Actually, in the construction in Section A of the map $\mathcal T$, this property was postulated, see Proposition~\ref{prop:uniqueops}. Let us now prove it in the framework of the alternative approach developed in Section C.
If Eq.~(\ref{eq:TARpsi}) is true, then Eq.~(\ref{eq:expvaluesR}) is a consequence of it along with \eqref{eq:scalarprod} and the fact that $\braket{\psi|H|\psi}$ is real for a Hermitian $H$. We have
\begin{equation}
\begin{split}
\mathcal T(A)\mathcal R(\ket\psi)
&=
[{\rm Re}(A)\otimes I^{(N)}
+{\rm Im}(A)\otimes J^{(N)}]
[{\rm Re}(\ket\psi)\otimes \ket{\psi_{\rm even}^{(N)}}
+{\rm Im}(\ket\psi)\otimes \ket{\psi_{\rm odd}^{(N)}}]
\\
&=
[{\rm Re}(A){\rm Re}(\ket\psi)
-
{\rm Im}(A){\rm Im}(\ket\psi)]
\otimes
\ket{\psi_{\rm even}^{(N)}}
+
[{\rm Re}(A){\rm Im}(\ket\psi)
+
{\rm Im}(A){\rm Re}(\ket\psi)]
\otimes
\ket{\psi_{\rm odd}^{(N)}}
\\
&={\rm Re}(A\ket\psi)
\otimes
\ket{\psi_{\rm even}^{(N)}}
+
{\rm Im}(A\ket\psi)]
\otimes
\ket{\psi_{\rm odd}^{(N)}}
\\
&=\mathcal R(A\ket\psi),
\end{split}
\end{equation}
i.e., the expectation values are preserved under the complex-to-real map.

\section{Formalization of Postulate (P4)}
\label{ap:pos_p4}

In this section we will provide a mathematical formulation of postulate (P4) and prove that indeed our construction fulfills it. First, we need to introduce some definitions about independent (or local) preparations of states and operations.

\begin{definition}[Independent preparation]
\label{def:indep-preparation}        

Let $\mathcal{H}_1$,\ldots,$\mathcal{H}_N$ be Hilbert spaces over a number field $\mathbb{F}$. We define an \textit{independent preparation} as a multi-linear embedding
\begin{equation}
\label{def:indep-prep}
\begin{aligned}
\xi_\mathbb{F}: \mathcal{H}_1\times\cdots\times\mathcal{H}_N&\to\mathcal{H},\\
(\ket{\psi_1},\ldots,\ket{\psi_N})&\mapsto\xi_\mathbb{F}(\ket{\psi_1},\ldots,\ket{\psi_N})
\end{aligned}
\end{equation}
\end{definition}

Note that, in the case $\mathbb{F}=\mathbb{C}$, i.e., in complex quantum mechanics, the embedding
\begin{equation}
\begin{aligned}
\xi_\mathbb{C}:\mathcal{H}_1\times\cdots\times\mathcal{H}_N&\to\mathcal{H}_1\otimes\cdots\otimes\mathcal{H}_N,\\(\ket{\psi_1},\ldots,\ket{\psi_N})&\mapsto\ket{\psi_1}\otimes\cdots\otimes\ket{\psi_N},
\end{aligned}
\end{equation}
corresponds to the standard way of formulating independent preparation -- the tensor product of states. Hence, Definition~\ref{def:indep-preparation} is indeed a generalization of independent preparation. Such generalization is also used in Ref.~\cite{4posare3}. Note that our construction does not contradict the results of Ref.~\cite{4posare3}, where it is shown that the tensor product of complex Hilbert spaces of the subsystems is the only possibility for the Hilbert space of a composite system under certain natural physical requirements. The analysis of Ref.~\cite{4posare3} essentially relies on the fact that orthogonal vectors are always distinguishable. As we mentioned previously, this is an inapplicable assumption for our real formulation of quantum mechanics as in our case the flag rotations do not change the physical state. 

\begin{definition}[Local operation]
\label{def:loc-op}          

Let $\mathcal{H}_1, \ldots ,\mathcal{H}_N$ be (finite-dimensional) Hilbert spaces over a field $\mathbb{F}$ and let $\mathcal{H}$ be the Hilbert space of the composite system over the same field $\mathbb{F}$. Let $\mathcal{L}(\mathcal{H}_i)$ be the space of linear operators in Hilbert space $\mathcal{H}_i$ and $A_i\in\mathcal{L}(\mathcal{H}_i)$. A \textit{local operation} acting on the subsystem associated with $\mathcal{H}_i$ and acting trivially everywhere else, is an embedding
\begin{equation}
\begin{aligned}
\eta_\mathbb{F}^i:\mathcal{L}(\mathcal{H}_i)&\to\mathcal{L}(\mathcal{H}),\\
A_i&\mapsto\eta_\mathbb{F}^i(A_i).
\end{aligned}
\end{equation}
\end{definition}

Note that, in the case $\mathbb{F}=\mathbb{C}$, i.e., in complex quantum mechanics, the embedding is
\begin{equation}
\begin{aligned}
\eta_\mathbb{C}^i:\mathcal{L}(\mathcal{H}_i)&\to\mathcal{L}(\mathcal{H}),\\
A_i&\mapsto\mathds{1}_1\otimes\cdots\otimes\mathds{1}_{i-1}\otimes A_i\otimes\mathds{1}_{i+1}\otimes\cdots\otimes\mathds{1}_N.
\end{aligned}
\end{equation}

\begin{definition}[Formal definition of postulate (P4)]
\label{def:pos-p4}      

For simplicity, we will consider a bipartite system, the generalization is straightforward.      Let $\mathcal{H}_1$ and $\mathcal{H}_2$ be Hilbert spaces over the field $\mathbb{F}$. Let $\mathcal{L}(\mathcal{H}_i)$ be the space of linear operators on the space $\mathcal{H}_i$. Let $A_i\in\mathcal{L}(\mathcal{H}_i)$, $\ket{\psi_1}\in\mathcal{H}_1$ and $\ket{\psi_2}\in\mathcal{H}_2$. Postulate (P4) means:
\begin{equation}
\label{eq:math-p4}
\begin{split}
\eta_{\mathbb F}^1(A_1)\xi_{\mathbb F}(\ket{\psi_1},\ket{\psi_2})&=\xi_{\mathbb F}(A_1\ket{\psi_1},\ket{\psi_2}),\\
\eta_{\mathbb F}^2(A_2)\xi_{\mathbb F}(\ket{\psi_1},\ket{\psi_2})&=\xi_{\mathbb F}(\ket{\psi_1},A_2\ket{\psi_2}).
\end{split}
\end{equation}
\end{definition}

Note that Eq. \eqref{eq:math-p4} in complex quantum mechanics reduces to
\begin{equation}
(A_1\otimes\mathds{1}_2)(\ket{\psi_1}\otimes\ket{\psi_2})=(A_1\ket{\psi_1})\otimes\ket{\psi_2},
\qquad
(\mathds{1}_1\otimes A_2)(\ket{\psi_1}\otimes\ket{\psi_2})=\ket{\psi_1}\otimes(A_2\ket{\psi_2}).
\end{equation}
Thus, complex quantum mechanics obeys postulate (P4). Therefore, Eq. \eqref{eq:math-p4} is indeed a generalization of the physical implication of the tensor product structure postulate. Namely, ``the operators used to describe measurements or transformations of system $A$ act trivially on $B$ and vice versa''.

The embeddings $\xi_\mathbb{R}$ and $\eta_\mathbb{R}^i$ in the real-valued formalism were actually defined in the main text and in sections B and C: $\xi_{\mathbb R}(\ket{\widetilde\psi},\ket{\widetilde\phi})=\ket{\widetilde\psi}\otimes_F\ket{\widetilde\phi}$, $\eta^1_{\mathbb R}(A)=A_1\otimes_F\mathds 1_2$, and $\eta^2_{\mathbb R}(A)=\mathds 1_1\otimes_F A_2$ for real vectors $\widetilde\psi$ and $\widetilde\phi$ and real operators $A_1$ and $A_2$, see Eqs.~(\ref{eq:otimesF}), (\ref{eq:otimesFpsiphi}), and (\ref{eq:AotimesFB}). Here $\mathds 1$ is the identity operator in the corresponding Hilbert space. The generalization to the multipartite case is obvious.   
These embeddings obey the following important properties:
\begin{equation}
\label{eq:def-indep-prep}
\xi_\mathbb{R}\colon\bigl(\mathcal{R}(\ket{\psi_1}),\ldots,\mathcal{R}(\ket{\psi_N})\bigr)\mapsto\mathcal{R}(\ket{\psi_1}\otimes\cdots\otimes\ket{\psi_N})
\end{equation}
for arbitrary complex vectors $\ket{\psi_1}$,\ldots, $\ket{\psi_N}$ and
\begin{equation}
\label{eq:def-loc-op}\eta_\mathbb{R}^i\colon{\mathcal{T}}(A_i)\mapsto{\mathcal{T}}(\mathds{1}_1\otimes\cdots\otimes\mathds{1}_{i-1}\otimes A_i\otimes\mathds{1}_{i+1}\otimes\cdots\otimes\mathds{1}_N)
\end{equation}
for any operator $A_i$ acting on a complex Hilbert space. Equalities (\ref{eq:def-indep-prep}) and (\ref{eq:def-loc-op}) give rise to the commutative diagrams depicted in Figs.~\ref{fig:com_diag_embed_states} and \ref{fig:com_diag_embed_op} showing how the complex and real pictures are related.

\begin{figure}[H]
\centering
\includegraphics{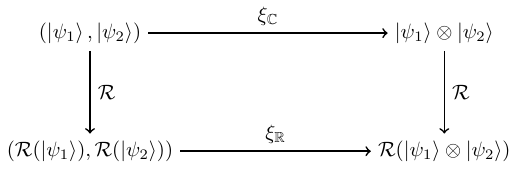}
\caption{Commutative diagram of the embeddings $\xi_\mathbb{C}$, $\xi_\mathbb{R}$, and map $\mathcal{R}$ for a bipartite system. Note that, by construction, $\mathcal{R}$ is an isomorphism and thus invertible. This allows us to move from the complex picture to the real as $\xi_\mathbb{C}=\mathcal{R}^{-1}\circ\xi_\mathbb{R}\circ\mathcal{R}$.}\label{fig:com_diag_embed_states}
\end{figure}
\begin{figure}[H]
\centering
\includegraphics{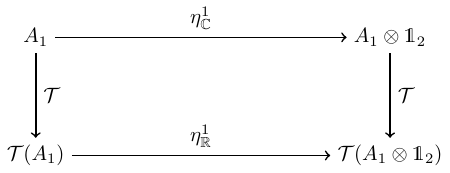}
\caption{Commutative diagram of the embeddings $\eta_\mathbb{C}^1$, $\eta_\mathbb{R}^1$, and map $\mathcal{T}$ for a bipartite system. Note that, by construction, $\mathcal{T}$ is an isomorphism and thus invertible.
This allows us to move from the complex picture to the real as $\eta_\mathbb{C}^1=\mathcal{T}^{-1}\circ\eta_\mathbb{R}^1\circ\mathcal{T}$.}\label{fig:com_diag_embed_op}
\end{figure}

Since $\mathcal T$ preserves the products of operators,
\begin{equation}
\label{eq:bipartite-prop}\eta_\mathbb{R}^1(\mathcal T(A_1))\eta_\mathbb{R}^2(\mathcal T(A_2))={\mathcal{T}}(A_1\otimes\mathds{1}_2){\mathcal{T}}(\mathds{1}_1\otimes A_2)={\mathcal{T}}(A_1\otimes A_2).
\end{equation}
For any two real operators $\widetilde A_1$ and $\widetilde A_2$ (generalization to the multipartite case is straightforward),
\begin{equation}
\eta_\mathbb{R}^1(\widetilde A_1)\eta_\mathbb{R}^2(\widetilde A_2)=
\widetilde A_1\otimes_F\widetilde A_2.
\end{equation}

\begin{proposition}
Postulate (P4) in form of Eq.~(\ref{eq:math-p4}) is fulfilled.
\end{proposition}
\begin{proof}

Consider the first expression in definition of Eq.~(\ref{eq:math-p4}) (the second one is analysed analogously) for real vectors $\mathcal R(\ket{\psi_1})$ and $\mathcal R(\ket{\psi_2})$ and a real operator $\mathcal T(A_1)$, where $\ket{\psi_1}$, $\ket{\psi_2}$ and $A_1$ are arbitrary complex vectors and an operator, respectively:

\begin{equation}
\label{eq:math-p4TR}
\eta_{\mathbb R}^1\big(\mathcal T (A_1)\big)\xi_{\mathbb R}\big(\mathcal R(\ket{\psi_1}),\mathcal R (\ket{\psi_2})\big)=\xi_{\mathbb R}\big(\mathcal T(A_1)\mathcal R(\ket{\psi_1}),\mathcal R(\ket{\psi_2})\big),\\
\end{equation}
Application of Eq.~(\ref{eq:TARpsi}) to the right-hand side of Eq.~\eqref{eq:math-p4TR} and substitution of embeddings \eqref{eq:def-indep-prep} and \eqref{eq:def-loc-op} into Eq.~\eqref{eq:math-p4TR} yields
\begin{equation}
\label{eq:our-math-p4}
\mathcal{T}(A_1\otimes\mathds{1}_2)\mathcal{R}(\ket{\psi_1}\otimes\ket{\psi_2})=\mathcal{R}(A_1\ket{\psi_1}\otimes\ket{\psi_2}).
\end{equation}
Thus, we need to prove Eq.~\eqref{eq:our-math-p4}. But this equality has already been proved for the general case (not necessarily a composite system), see again Eq.~(\ref{eq:TARpsi}). 
\end{proof}

\section{Density operators and reduced density operators in the real-valued formalism}
\label{ap:partial_trace}

Application of Eq.~(\ref{eq:RPi}) in the main text or Eq.~(\ref{eq:Roperator}) in the Supplementary Material to a complex-valued density operator $\rho$ gives the corresponding operator acting in a real Hilbert space:
\begin{equation}
\label{eq:Trho}
\mathcal T(\rho)={\rm Re}(\rho)\otimes I+{\rm Im}(\rho)\otimes J.
\end{equation}
According to Proposition~\ref{prop:herm2sym}, $\mathcal T(\rho)$ is a real symmetric operator, but the trace of $\mathcal T(\rho)$ is not one  as ${\rm tr}(I)=2$, ${\rm tr}(J)=0$, and ${\rm tr}[{\rm Re}(\rho)]={\rm Re}[{\rm tr}(\rho)]=1$.
In principle, although both observables and mixed states are operators, they are physically different objects (Especially, this is seen in the infinite-dimensional case, where a density operator must be a trace-class operator, in contrast to operators representing observables.) and can therefore be processed differently. Namely, for the density operators (mixed states), we can define the map (Note that in the end of Section C, we noticed that canonical representatives are defined up to a factor. Here, the subscript 1 refers to the property of trace preservation.)
\begin{equation}
\label{eq:T1rho}
{\mathcal T_1}(\rho)=\frac12{\rm Re}(\rho)\otimes I+\frac12{\rm Im}(\rho)\otimes J.
\end{equation}
This definition guarantees preservation of the trace, the relation ${\mathcal T}_1(\rho A)={\mathcal T}_1(\rho)\mathcal T(A)$, and, thus,
\begin{equation}
{\rm tr}(\rho A)=
{\rm tr}\big[
{\mathcal T}_1(\rho)
\mathcal T(A)
\big],
\end{equation}
i.e., invariance of the expectation values under the complex-to-real transformation, a generalization of Eq.~(\ref{eq:expvalues}) in the main text to the case of mixed states.

Now we will introduce a definition of reduced density operator in the real-valued formalism and check that it corresponds to the (real) partial trace. For simplicity, we will consider a bipartite system, the generalization is straightforward. 

\begin{definition}[Reduced density operator]
\label{def:partial_trace}
Let $\mathcal{H}_1$ and $\mathcal{H}_2$ be Hilbert spaces over a number field $\mathbb{F}$ corresponding to two quantum systems and $\mathcal H$ be the Hilbert space corresponding to the composite system. Let $A_2\in\mathcal{L}(\mathcal{H}_2)$ and consider a density operator in the composite system Hilbert space $\rho_{12}\in\mathcal{B}(\mathcal{H})$. Then $\rho_2$ is the reduced density operator of the system $\mathcal H_2$ if

\begin{equation}
\label{eq:reduceddensop}
\mathrm{tr}[\rho_{12}\eta_\mathbb{F}^2(A_2)]=\mathrm{tr}(\rho_2A_2),
\end{equation}
where  $\eta_\mathbb{F}^2$ was introduced in Definition \ref{def:loc-op}.
\end{definition}

For the complex case, $\mathcal H=\mathcal H_1\otimes\mathcal H_2$, $\eta^2_{\mathbb F}(A_2)=\mathds{1}\otimes A_2$, and it is well-known that the reduced density operator is given by the partial trace operation: $\rho_2={\rm tr}_1(\rho_{12})$, where the partial trace can be defined by

\begin{equation} 
\label{eq:parttr} \mathrm{tr}_1(\sigma_1\otimes\sigma_2)=\mathrm{tr}(\sigma_1)\sigma_2
\end{equation}
for arbitrary linear operators $\sigma_1$ and $\sigma_2$, with the continuation on arbitrary operators on $\mathcal H_1\otimes \mathcal H_2$ by linearity. We will show that the reduced operator in the real case is also given by the partial trace. Eq.~(\ref{eq:parttr}) does not depend on the number field, but in order for the partial trace to be well-defined in our construction, we need the representation of Eq.~(\ref{eq:Roperator}) for the real-valued operators for a composite system, which involves canonical representatives $I^{(N)}$ and $J^{(N)}$ of the equivalence classes from the tensor product of the flags of different subsystems (rather than Eq.~(\ref{eq:RPi}) from the main text, where the tensor-product structure of the two-dimensional flag operators $I_F$ and $J_F$ is not revealed).

We will show that the diagram depicted in Figure~\ref{fig:diag-ros} is commutative.

\begin{figure}[H]
\centering
\includegraphics{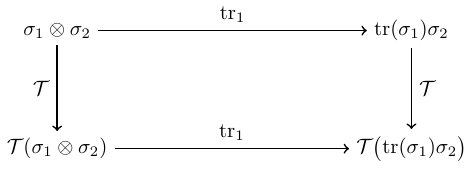}
\caption{Commutative diagram of the maps $\mathrm{tr}_1$ and $\mathcal{T}$. Here $\sigma_1$ and $\sigma_2$ are arbitrary complex linear operators.}
\label{fig:diag-ros}
\end{figure}
In this commutative diagram, we start with a density operator for the composite system in the complex case in the upper left part. We can obtain the reduced density operator in the real case in two ways. We can first map the joint density operator of two systems to the corresponding real-valued joint density operator and then apply the map to the reduced density operator in the real case. Or we can first apply the (complex) partial trace to obtain the reduced density operator for the complex case and then map it to the real-valued reduced density operator. The two results must be the same. This is true if the reduced density operator is given by the partial trace also in the real case. In short, ${\mathcal T}_1$ and the partial trace commute. For simplicity of notations, we use $\mathcal T$ instead of ${\mathcal T}_1$ since the normalization constant is unimportant here. 

Let us prove this commutative relation. We have (see Eq. \eqref{eq:IJN})
\begin{equation}
\mathcal T(\sigma_1\otimes\sigma_2)=
{\rm Re}(\sigma_1\otimes\sigma_2)\otimes I^{(2)}
+
{\rm Im}(\sigma_1\otimes\sigma_2)\otimes J^{(2)}.
\end{equation}
Since taking the real or imaginary part obviously commutes with the trace, 
\begin{equation}
{\rm tr}_1[\mathcal T(\sigma_1\otimes\sigma_2)]=
{\rm Re}[{\rm tr}(\sigma_1)\sigma_2]\otimes {\rm tr}_1(I^{(2)})
+
{\rm Im}[{\rm tr}(\sigma_1)\sigma_2]\otimes {\rm tr}_1(J^{(2)}).
\end{equation}
From the other side,
\begin{equation}
\mathcal T\big({\rm tr}(\sigma_1)\sigma_2\big)
=
{\rm Re}[{\rm tr}(\sigma_1)\sigma_2]\otimes I^{(1)}
+
{\rm Im}[{\rm tr}(\sigma_1)\sigma_2]\otimes J^{(1)}.     
\end{equation}
Hence, we need to prove that ${\rm tr}_1(I^{(2)})=I^{(1)}$ and ${\rm tr}_1(J^{(2)})=J^{(1)}$. Let us prove more general relations:

\begin{equation}
\label{eq:trIJ}
\mathrm{tr}_{1,\ldots,M}(I^{(N)})=I^{(N-M)},
\qquad
\mathrm{tr}_{1,\ldots,M}(J^{(N)})=J^{(N-M)},
\end{equation}
where $\mathrm{tr}_{1,\ldots,M}$ denotes the partial trace with respect to subsystems $1,\ldots,M$. We have 

\begin{align}
&\mathrm{tr}_i(\ket{\psi_{\mathrm{even}}^{(N)}}\bra{\psi_{\mathrm{even}}^{(N)}})
\nonumber
\\
&=\frac{1}{2^{N-1}}\sum_{\substack{k_1,\ldots,k_N=0
\\
\mathrm{even}\,|k|_{\rm H}\ }}\sum_{\substack{l_1,\ldots,l_N=0
\\
\mathrm{even}\,|l|_{\rm H} }}\sum_{\gamma=0}^1(-1)^{\frac{|k|_{\rm H}+|l|_{\rm H}}{2}}\ket{k_1,\ldots,k_{i-1}k_{i+1},\ldots,k_N}\bra{l_1,\ldots,l_{i-1}l_{i+1},\ldots,l_N}\braket{\gamma| k_i}\braket{l_i| \gamma}
\nonumber
\\
&=\frac{1}{2}(\ket{\psi_{\mathrm{even}}^{(N-1)}}\bra{\psi_{\mathrm{even}}^{(N-1)}}+\ket{\psi_{\mathrm{odd}}^{(N-1)}}\bra{\psi_{\mathrm{odd}}^{(N-1)}})=\frac{1}{2}I^{(N-1)},		
\end{align}
and similarly
\begin{align}
&\mathrm{tr}_i(\ket{\psi_{\mathrm{odd}}^{(N)}}\bra{\psi_{\mathrm{odd}}^{(N)}})
\nonumber
\\
&=\frac{1}{2^{N-1}}\sum_{\substack{k_1,\ldots,k_N=0 \\\mathrm{odd}\,|k|_{\rm H} }}^1\sum_{\substack{l_1,\ldots,l_N=0\\\mathrm{odd}\,|l|_{\rm H} }}^1\sum_{\gamma=0}^1(-1)^{\frac{|k|_{\rm H}-1+|l|_{\rm H}-1}{2}}\ket{k_1,\ldots,k_{i-1}k_{i+1},\ldots,k_N}\bra{l_1,\ldots,l_{i-1}l_{i+1},\ldots,l_N}\cdot\braket{\gamma| k_i}\braket{l_i| \gamma}
\nonumber
\\
&=\frac{1}{2}(\ket{\psi_{\mathrm{even}}^{(N-1)}}\bra{\psi_{\mathrm{even}}^{(N-1)}}+\ket{\psi_{\mathrm{odd}}^{(N-1)}}\bra{\psi_{\mathrm{odd}}^{(N-1)}})=\frac{1}{2}I^{(N-1)}.
\end{align}
Thus, by the definition of $I^{(N)}$ (see Eq. \eqref{eq:IJN}), we obtain the first equality in Eqs.~(\ref{eq:trIJ}). Analogously,

\begin{align}
& \mathrm{tr}_i(\ket{\psi_{\mathrm{odd}}^{(N)}}\bra{\psi_{\mathrm{even}}^{(N)}})
\nonumber \\
& =\frac{1}{2^{N-1}}\sum_{\substack{k_1,\ldots,k_N=0 \\\mathrm{odd}\,|k|_{\rm H} }}^1\sum_{\substack{l_1,\ldots,l_N=0\\\mathrm{even}\,|l|_{\rm H} }}^1\sum_{\gamma=0}^1(-1)^{\frac{|k|_{\rm H}-1+|l|_{\rm H}}{2}}\ket{k_1,\ldots,k_{i-1}k_{i+1},\ldots,k_N}\bra{l_1,\ldots,l_{i-1}l_{i+1},\ldots,l_N}\cdot\braket{\gamma| k_i}\braket{l_i| \gamma}
\nonumber \\
& =\frac{1}{2}(-\ket{\psi_{\mathrm{even}}^{(N-1)}}\bra{\psi_{\mathrm{odd}}^{(N-1)}}+\ket{\psi_{\mathrm{odd}}^{(N-1)}}\bra{\psi_{\mathrm{even}}^{(N-1)}})=\frac{1}{2}J^{(N-1)}.
\end{align}
Thus, by the definition of $J^{(N)}$ (see Eq. \eqref{eq:IJN}), we obtain the second equality in Eqs. \eqref{eq:trIJ}. This concludes the proof of the commutative diagram from Figure~\ref{fig:diag-ros}. 

\begin{remark}
In order to determine whether a pure state is entangled based on its real-valued form, one proceeds as in complex quantum mechanics, i.e. tracing out one subsystem and checking whether the reduced state is mixed or not. However, there is a subtle detail when handling density operators, as discussed in this section. Namely, when applying the map to density operators in order 
to obtain its real version, one has to include a factor of $\frac 12$ in order to preserve normalization (see Eq. \eqref{eq:T1rho}).

Therefore, given a real-valued density operator $\Tilde{\rho}$, if after tracing out one subsystem (for example $A$) the reduced state is pure (which, due to the factor of $\frac 12$, will translate into the condition $(\mathrm{tr}_A(\Tilde{\rho}))^2=\frac 12 \mathrm{tr}_A(\Tilde{\rho}) $), we conclude that the state $\Tilde{\rho}$ is separable. Consequently, we conclude that the state is entangled if the condition above does not hold.
\end{remark}

\section{Myrheim approach}
\label{ap:myrheim}

In this section, we compare our approach to the one of Ref.~\cite{Myrheim}. According to this approach, in our notations, the complex Hilbert space $\mathbb C^d$ also corresponds to the real Hilbert space $\mathbb R^d\otimes\mathbb R_F^2$. Consider now two real Hilbert spaces $\mathcal H_A=\mathbb R^{d_A}\otimes\mathbb R_{F_a}^2$ and $\mathcal H_B=\mathbb R^{d_B}\otimes\mathbb R_{F_b}^2$ and the tensor product for the composite system:

\begin{equation}
\mathcal{H}=\mathcal{H}_A\otimes\mathcal{H}_B.
\end{equation}
The following operators are introduced:
\begin{equation}
J_{A(B)}=\mathds{1}_{a(b)}\otimes         XZ_{F_a(F_b)}
\end{equation}
where subscript $a$ $(b)$ denote the ``main'' part of the real Hilbert space, i.e., $\mathbb R^{d_A}$ ($\mathbb R^{d_B}$), and $X$ and $Z$ are the Pauli matrices. 
This Hilbert space is decomposed as
\begin{equation}
\mathcal{H}=\mathcal{H}_+\oplus\mathcal{H}_-,
\end{equation}
where $\mathcal{H}_\pm=P_\pm \mathcal{H}$ and
\begin{equation}
P_\pm=\frac{1}{2}(\mathds{1}_{A}\otimes\mathds{1}_B\mp J_{A}\otimes J_{B}).
\end{equation}
Here the capital subscripts $A$ and $B$ refer to the whole spaces $\mathcal H_A$ and $\mathcal H_B$, i.e., including the flags.
It is argued that the complex tensor product space can be naturally identified with $\mathcal{H}_+$.
Explicitly, 
\begin{equation}
\begin{aligned}
P_+&=\frac{1}{2}\bigl(\mathds{1}_a\otimes\mathds{1}_{F_a}\otimes\mathds{1}_b\otimes\mathds{1}_{F_b}-\mathds{1}_a\otimes XZ_{F_a}\otimes\mathds{1}_b \otimes XZ_{F_b}\bigr)\\
&=\mathds{1}_a\otimes\mathds{1}_b\otimes\bigl(\ket{\psi_{\mathrm{even}}^{(2)}}\bra{\psi_{\mathrm{even}}^{(2)}}_{F_aF_b}+\ket{\psi_{\mathrm{odd}}^{(2)}}\bra{\psi_{\mathrm{odd}}^{(2)}}_{F_aF_b}\bigr).
\end{aligned}
\end{equation}
This shows that $P_+$ exactly corresponds to $P^\perp$, see Eq.~\eqref{eq:Pperp} in Section B in the bipartite scenario. 

In order to generalize to multipartite scenarios, a real Hilbert space of the form 
\begin{equation}
\mathcal{H}=\mathcal{H}_A\otimes\mathcal{H}_B\otimes\mathcal{H}_C,
\end{equation}
is considered. In addition, defined in analogy to $P_\pm$ is
\begin{equation}
Q_\pm=\frac{1}{2}(\mathds{1}_{A}\otimes\mathds{1}_C\mp J_{A}\otimes J_{C}).
\end{equation}
Then the physical subspace identified with the complex tensor product is $\mathcal{H}_{++}$ defined as
\begin{equation}
\mathcal{H}_{++}=P_+Q_+\mathcal{H}.
\end{equation}
Explicit computation shows that
\begin{equation}
P_+Q_+=\mathds{1}_a\otimes\mathds{1}_b\otimes\mathds{1}_c\otimes\bigl(\ket{\psi_{\mathrm{even}}^{(3)}}\bra{\psi_{\mathrm{even}}^{(3)}}_{F_aF_bF_c}+\ket{\psi_{\mathrm{odd}}^{(3)}}\bra{\psi_{\mathrm{odd}}^{(3)}}_{F_aF_bF_c}\bigr).
\end{equation}
This once again coincides with $P^\perp$ in the tripartite scenario. The generalization to an arbitrary number of subsystems is straightforward.

Myrheim's approach is similar to the one presented in this paper in that it also considers the tensor product of real Hilbert spaces. However, the key distinction lies in the treatment of the resulting space: Myrheim's approach involves restricting to a subspace, whereas this work considers the quotient space of the tensor product. This modification is crucial to preserve postulate (P1), which states that a pure state of a quantum system is represented by an arbitrary unit vector in a given Hilbert space, rather than a vector only from a subspace.

\section{Explicit calculation for Renou et al.'s model}
\label{ap:renou_experiment}
In this section we are going to review a model for an experiment proposed by Renou et al. in Ref.~\cite{Renou2021} and then describe it within the framework of quantum mechanics based on real numbers. We will see how, within this formalism, the experiment statistics from complex quantum mechanics are recovered.

\begin{figure}[h]
\centering
\includegraphics[scale=.6]{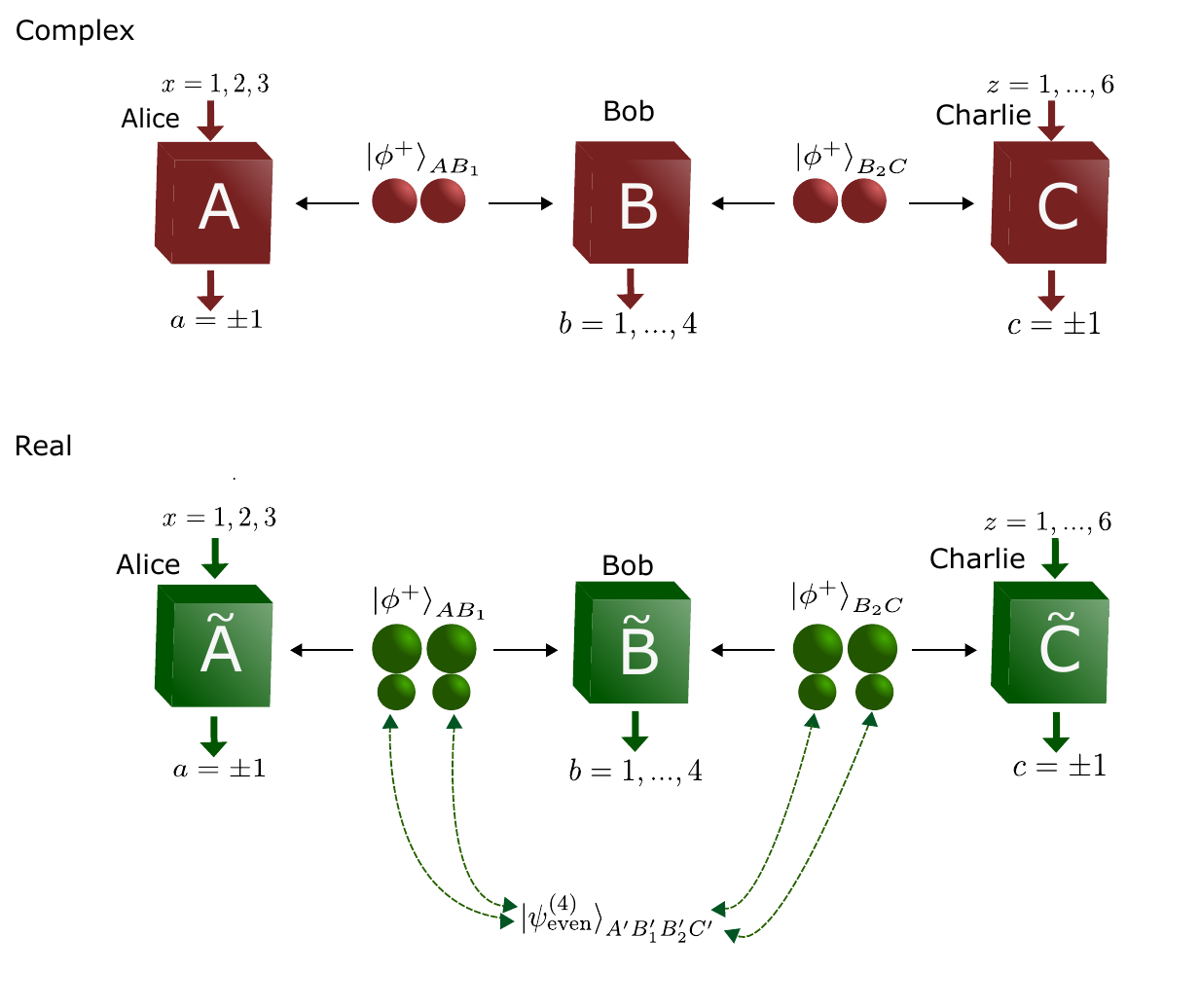}
\caption{\textbf{Adaptation of the figure by Renou et al. to the real-valued framework.} A maximally entangled state, namely $\ket{\phi^+}=\frac{1}{\sqrt{2}}(\ket{00}+\ket{11})$, is distributed to Alice and Bob and Bob and Charlie. 
Bob performs a Bell measurement, with outcomes labeled by $b$, that swaps the entanglement to Alice and Charlie, giving them a maximally entangled pair. 
In the complex case (top), Alice and Charlie perform three and six complex measurements labeled by $x$ and $z$, respectively, with two possible outcomes $\pm 1$. 
Afterwards, they gather the measurement statistics and observe a violation of a combination of CHSH inequalities. 
In the real case (bottom) we reproduce the measurement statistics by applying the mapping $\mathcal{R}$ (see Eq. \eqref{eq:mapstates}), which adds one flag state to each party (the primed systems represented by the small green spheres), to the states and operators. 
The same violation is achieved in this real version of quantum mechanics.}
\label{fig:scenario}
\end{figure} 

The authors introduce a network Bell type experiment  with an entanglement-swapping setup depicted on Figure~\ref{fig:scenario} (top part). The experiment involves three parties, Alice, Bob and Charlie. Two independent sources generate bipartite states: (i) for Alice and Bob, and (ii) for Bob and Charlie. Then Bob performs a fixed measurement with four outcomes. Alice chooses between three dichotomic measurements and Charlie chooses between six dichotomic measurements. Denote $x\in\{1,2,3\}$ and $z\in\{1,\ldots,6\}$ Alice's and Charlie's inputs (measurement choices) and $a,c\in\{-1,1\}$ their measurement outcomes. Bob's measurement outcome is denoted by $b=b_1b_2\in\{00,01,10,11\}$. This setup gives rise to a family of joint probability distributions $P(a,b,c|x,z)$ for parties' measurement outcomes for various inputs $x$ and $z$. That is, $P(a,b,c|x,z)\geq0$ and $\sum_{a,b,c}P(a,b,c|x,z)=1$ for all $x$ and $z$. Then, the following Bell quantity is constructed for this distribution.  We define
\begin{equation}
S_{xz}^b(P)=\sum_{a,c=\pm 1}a c\, P(a,b,c|x,z),	
\end{equation}
for every $b$, so that 
\begin{equation}
\begin{split}
\mathscr{T}_b(P)&=(-1)^{b_2}(S_{11}^b+S_{12}^b)+(-1)^{b_1}(S_{21}^b-S_{22}^b)\\
&+(-1)^{b_2}(S_{13}^b+S_{14}^b)-(-1)^{b_1+b_2}(S_{33}^b-S_{34}^b)\\
&+(-1)^{b_2}(S_{25}^b+S_{26}^b)-(-1)^{b_1+b_2}(S_{35}^b-S_{36}^b)
\end{split}
\label{functional}
\end{equation}
corresponds to the left-hand side of a particular Bell inequality resulting from the combination of three Clauser-Horne-Shimony-Holt (CHSH) inequalities~\cite{PhysRevA.93.040102, PhysRevA.98.042336}.
We combine all four inequalities into the term
\begin{equation}
\mathscr{T}(P)=\sum_b P(b)\mathscr{T}_b(P)
\end{equation}
with $P(b)=\sum_{a,c}P(a,b,c|x,z)$ being the probabilities of Bob's outcomes, which are independent of Alice's and Charlie's inputs.

The discussed question is which  probability distributions $P(a,b,c|x,z)$ can be reconstructed in complex and real quantum mechanics, respectively. In particular, which values of $\mathscr T(P)$ can be achieved in each model.
The interpretation within complex quantum mechanics is of the form
\begin{equation}
\label{eq:prob_distr_cqm}
P(a,b,c|x,z)=\sum_\lambda P(\lambda)\mathrm{tr}\bigl[(\rho_{AB_1}^\lambda\otimes\rho_{B_2C}^\lambda)(\Pi^A_{a|x}\otimes \Pi^B_b\otimes \Pi^C_{c|z})\bigr],
\end{equation}
for some probability distribution $P(\lambda)$ of a shared randomness (or a ``hidden variable'') $\lambda$, some quantum states (distributed by the sources) $\rho_{AB_1}^\lambda$ and $\rho_{B_2C}^\lambda$, and some projective measurement operators $\Pi^A_{a|x}$, $\Pi^B_b$ and $\Pi^C_{c|z}$ with $\sum_a \Pi^A_{a|x}=\mathds 1_A$ for all $x$, $\sum_b \Pi^B_{b}=\mathds 1_B$, and $\sum_c \Pi^C_{c|z}=\mathds 1_C$ for all $z$. 

The following experiment in the framework of the usual complex quantum mechanics is suggested, generating a family of probability distributions $\bar P(a,b,c|x,y)$. Both sources distribute two-qubit Bell states $\ket{\phi^+}$, i.e., $\bar\rho_{AB_1}=\bar\rho_{B_2C}=\ket{\phi^+}\bra{\phi^+}$.  Shared randomness is not used. Then Bob performs a Bell measurement on his two qubits (one qubit from each source). The corresponding projectors are
\begin{equation}
(\bar \Pi^B_{00},\bar\Pi^B_{01},\bar\Pi^B_{10},\bar\Pi^B_{11}) = \left(\ket{\phi^+}\bra{\phi^+},\, \ket{\psi^+}\bra{\psi^+},\, \ket{\phi^-}\bra{\phi^-},\, \ket{\psi^-}\bra{\psi^-}\right), \label{eq:bellmeasurement}
\end{equation}
where $\ket{\psi^\pm}=\frac{1}{\sqrt{2}}(\ket{10}\pm \ket{01})$ and $\ket{\phi^\pm}=\frac{1}{\sqrt{2}}(\ket{00}\pm \ket{11})$ are the four Bell states. Note that, following \cite{Renou2021}, we use a different convention for $\ket{\psi^-}$ in comparison with the standard one. For each of Bob's outcomes the marginal state that Alice and Charlie share is one of the four Bell states.

Alice's and Charlie's particles end up maximally entangled through Bob's measurement and never interacting with each other. Moreover, the protocol can be interpreted as teleportation of Bob's part of the entangled pair shared with Alice to Charlie. As shown in \cite{teleport}, by making use of the no-signaling principle, it is concluded that reliable teleportation of a $d$-dimensional state requires a classical channel of $2\log_2(d)$ bits of capacity. 
Thus, two classical bits of information (coming from Bob's measurement) allow for at most one maximally entangled qubit-pair (``e-bit'').

Alice  measures the observables (given by Hermitian operators and, hence, the corresponding projector-valued resolutions of unity: $A_x= \bar\Pi^A_{1|x}-\bar\Pi^A_{-1|x}$)
\begin{equation}
A_1 = Z, \qquad A_2 = X, \qquad A_3 = Y,
\end{equation}
and Charlies measures the observables 
\begin{eqnarray}\label{Eq-CharlyMeas}
\begin{aligned}
C_1=&\frac{X+Z}{\sqrt{2}}, &\quad C_2 =& \frac{Z-X}{\sqrt{2}}, &\quad C_3 =& \frac{Y+Z}{\sqrt{2}},\\
C_4 =& \frac{Z-Y}{\sqrt{2}},&\quad C_5 =& \frac{X+Y}{\sqrt{2}}, & C_6=&\frac{X-Y}{\sqrt{2}},
\end{aligned}
\end{eqnarray}
with the outcomes $a,c\in\{-1,1\}$, respectively, where $X$, $Y$ and $Z$ again are the Pauli matrices. Analogously, $C_x= \bar\Pi^C_{1|x}-\bar\Pi^C_{-1|x}$.

Thus, we obtain the probability distributions $\bar P(a,b,c|x,y)$, in Eq.~(\ref{eq:prob_distr_cqm}) by construction:
\begin{equation}
\label{eq:prob_distr_cqm_bar}
\bar P(a,b,c|x,z):=\mathrm{tr}\bigl[(\bar\rho_{AB_1}\otimes\bar\rho_{B_2C})(\bar\Pi^A_{a|x}\otimes \bar\Pi^B_b\otimes \bar\Pi^C_{c|z})\bigr].
\end{equation}
Here $\bar P(a,b,c|x,z)$ denotes the probability distribution which is obtained for the states and measurements defined above (Eqs.~\eqref{eq:bellmeasurement}--\eqref{Eq-CharlyMeas}).  Then, 
$\bar P(b)=\frac{1}{4}$ and $\mathscr T_b(\bar P)=6\sqrt2$ for all outcomes $b$ and, thus,
\begin{equation}
\mathscr{T}(\bar P)=6\sqrt{2}\approx8.4853.
\end{equation} 
On the other side, Renou et al. derive the upper bound $\mathscr{T}( P)\leq 7.6605$ if all operators in Eq.~(\ref{eq:prob_distr_cqm}) are real.

However, Eq.~(\ref{eq:prob_distr_cqm}) reflects the embedding of a subsystem into a composite system as the tensor product. In Section D, we defined the embedding for operators into the quotient space. Then, instead of Eq.~\eqref{eq:prob_distr_cqm}, in the real-valued formalism we have the form (see Eqs.~\eqref{eq:otimesF}, \eqref{eq:otimesFpsiphi}, and \eqref{eq:AotimesFB} for the definition of the ``flag tensor product'' $\otimes_F$)
\begin{equation}
\label{eq:Prealform}
P(a,b,c|x,z)=\sum_\lambda P(\lambda)\mathrm{tr}\bigl[(\rho_{AB_1}^\lambda\otimes_F\rho_{B_2C}^\lambda)(\Pi^A_{a|x}\otimes_F \Pi^B_b\otimes_F \Pi^C_{c|z})\bigr].       
\end{equation}

We put by definition

\begin{equation}
\label{eq:prob_distr_cqm_bar_r}
\begin{split}
\bar P_{\rm real}(a,b,c|x,z)&:=\mathrm{tr}\bigl\{[{\mathcal T}_1(\bar\rho_{AB_1})\otimes_F{\mathcal T}_1(\bar\rho_{B_2C})][\mathcal T(\bar\Pi^A_{a|x})\otimes_F \mathcal T(\bar\Pi^B_b)\otimes_F \mathcal T(\bar\Pi^C_{c|z})]\bigr\}
\\
&=
\braket{\widetilde\psi_0|\mathcal T(\bar\Pi^A_{a|x})\otimes \mathcal T(\bar\Pi^B_b)\otimes \mathcal T(\bar\Pi^C_{c|z})|\widetilde\psi_0},
\end{split}
\end{equation}
where the map $\mathcal T_1$ for the density operators was defined in Eq.~\eqref{eq:T1rho}. 
Also see the discussion in the end of Section C that $\otimes_F$ can be replaced by the usual tensor product $\otimes$ if they act on canonical representatives of the equivalence classes. Finally, 
\begin{equation}
\ket{\widetilde\psi_0}=\mathcal{R}\big(\ket{\phi^+}_{AB_1}\otimes\ket{\phi^+}_{B_2C}\big)=\ket{\phi^+}_{AB_1}\otimes\ket{\phi^+}_{B_2C}\otimes\ket{\psi_{\mathrm{even}}^{(4)}}_{A'B_1'B_2'C'},
\end{equation}
where the primed subindices refer to the flag state of the corresponding particle.
See Figure~\ref{fig:scenario} (bottom part) for an illustration of the setup. Note that the first line of Eq.~\eqref{eq:prob_distr_cqm_bar_r} clearly shows the independence of sources. The equality with the second line of Eq.~\eqref{eq:prob_distr_cqm_bar_r}, which formally is not source-factorizable in the usual sense, is due to the introduced equivalence relations. This reconciles  the independence of the sources with the apparent non-factorability of the state. See also the main text.

Of course, by the general construction, $\bar P_{\rm real}=\bar P$, but
let us explicitly calculate the Bell quantity $\mathcal T(\bar P_{\rm real})$. We have

\begin{equation}
\mathcal{T}(\bar\Pi_b) = (\bar\Pi_b^B)\otimes \mathds{1}_{B'}.
\end{equation}
Since every operator so far is real also in the complex description of the experiment, the state of the system after the measurement conditioned on Bob's result $b$ is
\begin{equation}
\begin{aligned}
\tilde\rho_{ABC}^{(b)} =&\frac{1}{P(b)} (\mathds{1}_A\otimes \bar\Pi^B_b\otimes \mathds{1}_C\otimes\mathds{1}_{A'B_1'B_2'C'})\ket{\widetilde\psi_0}\bra{\widetilde\psi_0} (\mathds{1}_A\otimes \bar \Pi^B_b\otimes \mathds{1}_C\otimes\mathds{1}_{A'B_1'B_2'C'})\\
=& \ket{\phi^+}\bra{\phi^+}_{AC} \otimes \bar \Pi^B_b\otimes \ket{\psi_{\mathrm{even}}^{(4)}}\bra{\psi_{\mathrm{even}}^{(4)}}_{A'B_1'B_2'C'}.
\end{aligned}
\end{equation}
Let us consider Bob's outcome $b=b_1b_2=00$ (the calculation for the other values of $b$ are identical). The marginal state of Alice and Charlie after Bob's measurement is obtained by tracing out subsystems $B_1,B_2,B_1',B_2'$ and reads
\begin{equation}
\label{postmeasstate}
\tilde\rho_{AC}^{(00)}=\mathrm{tr}_{B_1B_2B_1'B_2'}\big(\tilde\rho_{ABC}^{(00)}\big)=\ket{\phi^+}\bra{\phi^+}_{AC}\otimes\frac{1}{2}\left(\ket{\phi^-}\bra{\phi^-}_{A'C'}+\ket{\psi^+}\bra{\psi^+}_{A'C'}\right),
\end{equation}
since, recall that, the flag states for two parties are given by $\ket{\psi_{\mathrm{even}}^{(2)}}=\ket{\phi^-}$ and $\ket{\psi_{\mathrm{odd}}^{(2)}}=\ket{\psi^+}$. Thus,
\begin{equation}
\bar P_{\mathrm{real}}(a,b,c|x,y) = \frac{1}{4}\mathrm{tr}\left[\tilde\rho_{AC}^{(b)} (\mathds{1}_{AA'}+a\, \mathcal{T}(A_x)_{AA'})\otimes (\mathds{1}_{CC'}+ c\, \mathcal{T}(C_y)_{CC'})\right].
\end{equation}

According to the definition in Eq.~\eqref{functional} for $\mathscr{T}_b(P)$, we get 
\begin{equation}
\label{functional00}
\mathscr{T}_{00}(\bar P_{\mathrm{real}})=(S_{11}^{00}+S_{12}^{00})+(S_{21}^{00}-S_{22}^{00})
+(S_{13}^{00}+S_{14}^{00})-(S_{33}^{00}-S_{34}^{00})
+(S_{25}^{00}+S_{26}^{00})-(S_{35}^{00}-S_{36}^{00}),
\end{equation}
where
\begin{equation}
S_{11}^{00}=\mathrm{tr}\left(\mathcal{T}(A_1)_{AA'}\mathcal{T}(C_1)_{CC'}\tilde\rho_{AC}^{(00)}\right)=\mathrm{tr}\left[\left(Z_A\otimes\mathds{1}_{A'}\otimes\frac{1}{\sqrt{2}}(X_C+Z_C)\otimes\mathds{1}_{C'}\right)\tilde\rho_{AC}^{(00)}\right]=\frac{1}{\sqrt{2}}.
\end{equation}
In this case, every operator was already real in the original description. Let us give an example of a computation with operators containing complex numbers:
\begin{equation}
\begin{aligned}
S_{33}^{00}=&\mathrm{tr}\left(\mathcal{T}(A_3)_{AA'}\mathcal{T}(C_3)_{CC'}\tilde\rho_{AC}^{(00)}\right)\\
=&\mathrm{tr}\left[\left(\mathrm{Im}(Y)_A\otimes X_{A'}Z_{A'}\otimes\frac{1}{\sqrt{2}}\big(\mathrm{Im}(Y)_C\otimes X_{C'}Z_{C'}+Z_C\otimes\mathds{1}_{C'}\big)\right)\tilde\rho_{AC}^{(00)}\right]=-\frac{1}{\sqrt{2}},
\end{aligned}
\end{equation}
since, recall that, $J=XZ$, where $X$ and $Z$ are the Pauli matrices. Also note that the product $X_{A'}Z_{A'}\otimes X_{C'}Z_{C'}$ and $-I_{A'C'}$ have equal effects on vectors of the form $\ket{\widetilde\psi}\otimes_F\ket{\widetilde\phi}$, and $\rho^{(00)}_{AC}$ is a convex combination of vectors of this form.
Analogously, one can compute the expectation value for every combination in Eq. \eqref{functional00}, yielding
\begin{equation}
\begin{aligned}
S^{00}_{11}=&\frac{1}{\sqrt{2}}, &\quad S^{00}_{12}=&\frac{1}{\sqrt{2}}, &\quad  S^{00}_{21}=&\frac{1}{\sqrt{2}}, &\quad  S^{00}_{22}=&-\frac{1}{\sqrt{2}}, &\quad\\
S^{00}_{13}=&\frac{1}{\sqrt{2}}, &\quad  S^{00}_{14}=&\frac{1}{\sqrt{2}}, &\quad  S^{00}_{33}=&-\frac{1}{\sqrt{2}}, &\quad  S^{00}_{44}=&\frac{1}{\sqrt{2}}, &\quad\\
S^{00}_{25}=&\frac{1}{\sqrt{2}}, &\quad  S^{00}_{26}=&\frac{1}{\sqrt{2}}, &\quad  S^{00}_{35}=&-\frac{1}{\sqrt{2}}, &\text{ and }  S^{00}_{36}=&\frac{1}{\sqrt{2}}. &\quad
\end{aligned}
\end{equation}
We recover all expectation values predicted by complex quantum mechanics and achieve $\mathscr T_{00}(\bar P_{\mathrm{real}})=6\sqrt{2}$. Similar calculations of $\mathscr T_b(\bar P_{\rm real})$ for all other $b$ lead to the same final Bell quantity 
\begin{equation}
\mathscr{T}(\bar P_{\mathrm{real}})=6\sqrt{2}.
\end{equation}
Therefore, observing this value in the experiment does not distinguish between complex and real quantum mechanics.

\end{document}